\documentclass{amsart}          

\usepackage{a4wide}            
\usepackage{amssymb,latexsym,amsmath,amsthm}  
\usepackage{epsfig}
\usepackage{hyperref}
\usepackage{url}

\DeclareMathOperator{\lm}{lm}

\DeclareMathOperator{\NF}{NF}

\DeclareMathOperator{\Mon}{Mon}

\newcommand{\w}{{\omega}}
\newcommand{\R}{{\mathbb{R}}}
\newcommand{\N}{{\mathbb{N}}}
\newcommand{\K}{{K}}
\newcommand{\G}{\mathbb{G}}

\newcommand{\C}{{\mathbb{C}}}
\newcommand{\Q}{{\mathbb{Q}}}
\newcommand{\Z}{{\mathbb{Z}}}

\renewcommand{\d}{\partial}
\newcommand{\tri}{\triangle}

\newcommand{\ARA}{\begin{array}}
\newcommand{\ARE}{\end{array}}
\newcommand{\SMA}{\begin{smallmatrix}}
\newcommand{\SME}{\end{smallmatrix}}
\newcommand{\BMA}{\begin{matrix}}
\newcommand{\BME}{\end{matrix}}

\newtheorem{lemma}{Lemma}[section]
\newtheorem{theorem}[lemma]{Theorem}
\newtheorem{definition}[lemma]{Definition} 
\newtheorem{example}[lemma]{Example}
\newtheorem{remark}[lemma]{Remark} 

\newcommand{\ra}{\rightarrow}

\begin{document}
\thispagestyle{empty}
\title[Symbolic Finite Difference Schemes]{A Symbolic Approach to Generation and Analysis of Finite Difference Schemes of Partial Differential Equations}

\author[V. Levandovskyy]{Viktor Levandovskyy}

\author[B. Martin]{Bernd Martin}

\address[V. Levandovskyy]{Lehrstuhl D f\"ur Mathematik, RWTH Aachen, Templergraben 64, 52062 Aachen, Germany}
\email[V. Levandovskyy]{Viktor.Levandovskyy@math.rwth-aachen.de}

\address[B. Martin]{Lehrstuhl Algebra und Geometrie, TU Cottbus, Postfach 10 13 44, 03013 Cottbus, Germany}
\email[B. Martin]{martin@math.tu-cottbus.de}

\maketitle

\tableofcontents

\begin{abstract}
In this paper we discuss three symbolic approaches for the generation of
a finite difference scheme of a partial differential equation (PDE). We prove,
that for a linear PDE with constant coefficients these three approaches are
equivalent and discuss the applicability of them to nonlinear PDE's
as well as to the case of variable coefficients. Moreover, we systematically
use another symbolic technique, namely the cylindrical algebraic decomposition, in order to derive the conditions on the von Neumann stability of a difference scheme for a linear PDE with constant coefficients. For stable schemes we demonstrate algorithmic and symbolic approach to handle both continuous and discrete dispersion. We present an implementation of tools for generation of schemes, which rely on Gr\"obner basis, in the system \textsc{Singular} and present numerous examples, computed with our implementation. In the stability analysis, 
we use the system \textsc{Mathematica} for cylindrical algebraic decomposition.
\end{abstract}

\section*{Introduction}

The finite difference method for linear PDE's
belongs to the very classical topics in mathematics. However,
its exposition in the classical books like \cite{T} often 
contains mysterious steps, relying on the huge experience,
gathered in last centuries. An algebraist is often confused
with such exposition and asks, whether there is a way to
split the whole picture into purely analytic and algebraic
parts and how is it possible to automatize the process of
scheme generation and further analysis of its properties. 
Hence the ideas to generate finite difference scheme
in an algebraic (or a symbolic) way are folklore, see 
for instance \cite{GV96, FIDE} for approaches and older implementations.

Terminologically, we address a \textit{difference scheme} as symbolic polynomial expression involving unknown function and shift (or difference) operators. We do not attach initial and/or boundary conditions to differential resp. difference 
equations, since the generation of a difference scheme as above is independent on them. As we will see, von Neumann stability can be seen as global
result, always being a necessary condition for stability of a problem with
initial and/or boundary conditions (and sometimes sufficient condition as well). Of course, one uses initial and/or boundary conditions for numerical solving, but the splitting of the whole problem into purely symbolic pre-processing 
and numerical post-processing seems to be the way to address such problems in the future.

In the article \cite{BGM}, Gerdt et al. used for the first time several new
ideas like the use of integral relations in discrete form (especially useful
if one deals with conservation laws), the formulation
of the scheme generation problem as a task for difference elimination
and the systematic use of involutive and Gr\"obner bases. 
Inspired by these ideas, we present our approaches, which will make 
the overall picture of scheme generation and analysis more complete.

The ideas of algebraic analysis suggest the separation of
a problem into analytic and algebraic parts. This allows,
in the case of linear PDE's, to treat systems of equations
via modules over algebras ($D$-module theory, homological algebra etc.).
In such a case there exist many algorithms and several powerful
implementations. Gr\"obner and involutive bases play a fundamental role in
such algorithms, see e.~g. \cite{SST, JFP, CQR}.

In the other direction, the differential algebra (e.~g. \cite{Ritt}) and difference algebra
(e.~g. \cite{CohnRM}) theories allow one to tackle nonlinear equations as well, though the algorithms in these realms are very complicated. In particular,
up to now we do not know any implementation of a basis construction
algorithm for the difference algebra. Notably, a new algorithmic 
approach to (nonlinear) difference equations seem to follow from
the \textit{letterplace} approach \cite{LL09}. However, one needs to
elaborate all details of this promising direction.

The usage of the famous cylindrical algebraic decomposition (CAD) originated from real algebraic geometry to von Neumann stability problems goes back to
\cite{FIDE, HLS}. Since that times more implementations of the CAD 
evloved and their performance has been greatly enhanced.

This paper is organized as follows.
We start with minimal prerequisites and revisit the basic concepts
of scheme generation, paying attention to the algebraic background
including Gr\"obner bases and elimination tools in the Section \ref{s1}.

We discuss three symbolic methods, used in applications
in the Section \ref{s3} and prove their equivalence in the case of a linear
PDE with constant coefficients in the Theorem \ref{MainT}. In cases
of other PDE's only one method will work in general, see Remark \ref{why3}.
As for linear PDE's as in the Theorem, we propose to use module
formulation and Gr\"obner bases for eliminating module components,
which can be seen as a natural generalization of the Gaussian elimination
to matrices over rings. We show the merits of this method, applied
to the classical equations of mathematical physics (heat, wave, advection equations) 
for various approximations. Note, that the method we propose can handle high 
order approximations, which are seldom used in the theory of PDE, but quite often 
in the theory of ODE as high order Runge-Kutta methods.

In the Section \ref{sStability} we present an algebraic and constructive 
formulation of von Neumann stability via ring homomorphism. We
shortly revisit the concepts of cylindrical algebraic decomposition
and connect its use to the questions, arising from difference schemes.

In the Section \ref{sExWave} we consider the $\lambda$-wave equation $u_{tt} = \lambda u_{xx}$ and perform both generation and stability analysis of several difference schemes,
obtained with different approximations. We demonstrate the merits of the semi-factorized
form of a difference scheme, it turns out to be especially useful for higher
dimensional situation.

In the Section \ref{sDispersion} we demonstrate, that the determination of continuous respectively discrete dispersions for a PDE respectively its difference scheme can be algebraized to the large extent as well.

All the examples in this paper have been computed with our implementation
of tools for difference scheme in a freely available computer algebra
system \textsc{Singular} \cite{Singular}. The corresponding library
\texttt{findifs.lib} will be distributed with the next version
of \textsc{Singular}. For the cylindrical algebraic decomposition
we use the commercial system \textsc{Mathematica}; indeed there
are freely available systems like \textsc{QEPCAD} and \textsc{REDLOG},
which are able to do the decomposition as well.

From the viewpoint of applied mathematics there is a general skepticism about the use
of symbolic methods. With this paper we want to stimulate a discussion between scientists of both fields based on a realistic viewpoint.

\section{Algebraization of differential and difference equations}
\label{s1}


\subsection{Types of operator algebras}

First of all we have to fix a computable field $k$, our base field, 
it is mostly the field of rational numbers $\Q$ or complex rational numbers
$\Q[i]$. (Computing with reals or complex numbers is in principle possible,
but only with a fixed precision, 
i.e. with a rational approximation, or one can
compute with algebraic extensions in  
roots of polynomials \-- this being not so interesting for our purpose.)
We can extend the base field by indeterminate constants, i.e. rational functions
in the constants: $\K=k(a,b,c,\ldots)$, which can be specialized 
to special values at any step of our computation if necessary. 
$\K$ is called the field of constants.

We fix a set of variables $x:=(x_1,\ldots , x_n)$ and an 'algebra' of functions
$C=C(x)$ in the $x's$, for instance differentiable functions or functions in discrete (shiftable) arguments. 
$C$ is not our object of computation. 
Instead we consider various operator algebras, consisting of operators, which act on $C$.
There are many operators, which one can handle in this framework, for example
\begin{enumerate}
\item multiplication with a variable: $\ x_i: u(x)\in C\mapsto x_iu(x)\in C$;
\item multiplication with a function 
$f\in C$: $\ m_f:u(x)\in C\mapsto f(x)\cdot u(x)\in C$.
\item partial differentiation: $\ \partial_i: u(x)\in C\mapsto \tfrac{\partial u(x)}{\partial x_i}\in C$;
\item partial shift operators: 
$\ T_i: u(x)\in C\mapsto u(x_1,\ldots ,x_i+1,\ldots ,x_n)\in C$;
\item partial $\lambda$-shift operators: 
$\ T^{\lambda}_{i}: u(x)\in C\mapsto u(x_1,\ldots ,x_i+\lambda,\ldots ,x_n)\in C$, 
clearly $T_i = T^1_i$;
\item partial difference operators: 
$\ \Delta_i = T^{\lambda}_{i} - 1;$ \ $ u(x)\in C \mapsto T^{\lambda}_{i}(u(x)) - u(x) \in C$;
\item $q$-dilation operators: $\ D_d:u(x)\in C
	\mapsto u(x_1,\ldots ,qx_i,\ldots ,x_n)\in C$,
\item et cetera.
\end{enumerate}
Fix a set $S$ of operators, we consider the operator algebra 
$A:=\K\langle S\rangle $ being the
sub algebra of all (linear) operators $Hom_\K(C,C)$ generated by $S$, 
i.e. the smallest 
linear subspace closed under multiplication of operators.
As long as $S$ consists of a finite number of pairwise commuting and independent 
operators the resulting algebra is isomorphic to a polynomial ring: 
$\K[t_1,\ldots , t_m]$.
Otherwise we get (non-commutative) quotient algebra of the free algebra
$\K\langle S\rangle $ by the two-sided ideal of all relations of $S$.

\begin{example}[\textbf{Algebras with constant coefficients}]
\label{ConstCoef}

The algebras of linear partial differential and shift (or difference) operators with 
constant coefficients are commutative $\K$--algebras, isomorphic to 
$\K[x_1,\dots,x_n]$. We denote them by
 $\K[\d_1,\dots,\d_n]$ and $\K[T_1,\dots,T_n]$ respectively.
\end{example}

\begin{example}[\textbf{Algebra with polynomial coefficients}]
\label{PolyCoef}

The algebra of linear partial differential operators with \textbf{polynomial} 
coefficients is the Weyl algebra. It is non--commutative but has simple 
commuting relations. We denote this algebra as $\K\langle x,\partial 
\mid \partial x=x\partial +1 \rangle$, what means that 
this algebra is generated by the $\{x,\partial\}$ over the field $\K$, 
and has the multiplication, defined on its generators: $x \cdot \partial = x 
\partial, \partial  \cdot x = x\partial +1$. The generalization to 
multivariate case is easy, the variable $\{x_i\}$ commutes with the variable 
$\{\d_j\}$ expect for the case $j=i$, then the relation as above applies.

Why the variables satisfy this relation? Consider the multiplication of
two operators, $x:=x_i$ and $\partial:=\frac{\d}{\d x_i}$. Take some
differentiable function $f$ and apply the Leibnitz rule to the product: 
$(\d x)(f) = \d(xf) = x \d(f) + f = (x \d +1 )(f)$. Hence, 
in the operator form $(\d x - x \d -1)(f)=0$.

Consider the algebra of linear $\lambda$-shift operators with \textbf{polynomial} 
coefficients, having in mind $\lambda = \tri x$. Like before, we ca derive the 
relation between operators $T:=T^{\lambda}$ and $x$. 
For any function in discrete argument $f$, $(T x)(f) = T(xf) = T(x)T(f) = (x + \tri x) T(f) = (xT)(f) + \tri x \cdot T(f)$. 
Hence, in the operator form this relation becomes $Tx = xT + \tri x \cdot T$.
The algebra, corresponding to the difference operator $\Delta= T^{\tri x} - 1$
has the relation $\Delta x = x \Delta + \Delta + 1$. These algebras are so called $G$-algebras, in which Gr\"obner basis algorithms exist and are implemented in the system \textsc{Singular:Plural} (\cite{Plural}), see e.~g. \cite{LVdiss}.
\end{example}

\begin{example}[\textbf{Algebra with coefficients in rational functions}]
\label{RatCoef}

Algorithmic computations are possible in the algebras whose coefficient fields are rational functions in $x$: $\K(x)\langle \partial \mid \partial 
x=x\partial +1 \rangle$ and $\K(\tri x,x)\langle T \mid Tx = xT + 
\tri x \cdot T  \rangle $, which are called \textbf{rational} Weyl 
algebra resp. \textbf{rational} shift algebra. 
Algebraically speaking, a passage from polynomial algebra to a rational 
algebra may be achieved by means of localization.
\end{example}
\begin{example}[
\textbf{Differential and Difference Algebra}]

In order to handle non-linear differential resp. difference equations with polynomial nonlinearities, one can consider a full
differential resp. difference 
algebra $\K[\{ O^\beta u \mid \beta\in\N^n\}]$, where
$O$ stands for differential resp. difference operators. Note,
that $O^\beta u$ is a variable, representing $O^\beta(u)$,
where $u=u(x)$ is a dependent variable and $x$ an independent one. 
Note, that these algebras are commutative and their infinite generating sets
are algebraically independent (but dependent differentially resp. in
a difference way).

The given nonlinear equations can be taken as generators of the differential  resp. difference ideal (that is an ideal, closed under the action of corresponding operators) in an above algebra. 

Since such algebras are infinitely generated, they are not Noetherian. 
Gr\"obner basis-like algorithms are therefore not terminating in general.
Nevertheless some parts of the theory from the linear situation 
is extended to this general situation.
\end{example}

In this paper we work algorithmically with linear partial differential operators with constant coefficients. However, in some parts we address and discuss
more general situations as well.

\subsection{Presentation of a system of differential equations}
Any partial differential equation (depending of its kind) defines an element of 
an algebra of corresponding differential operators. A solution of a system of equations
fulfill any equation of the left ideal generated by the equations 
of the system in the algebra. Hence, the solution does not depend on the choice of a basis (that is, a generating set) of the ideal. 
The first possibility of applying symbolic algebra is to compute a better basis
of the ideal, like a {\em Gr\"obner basis} or an {\em involutive basis} (like {\em Janet basis}, 
see \cite{JFP,VPG,WS} - as far as it is possible.
The advantage of thus a pre-processing could be: check the consistency of the system of equations, find hidden constraints or integrability conditions of the system, 
determine the dimension of the solution space etc. 

These data are well known for standard equations from mathematical physics, but
the methods we propose are methodologically applicable to any system of equations.
Let us recall a small example (by W.~Seiler \cite{WS}) as an illustration.
\begin{example}

$$ \left\{\begin{array}{l} u_z+y\,u_x=0 \\
u_y=0\end{array}\right. \Longrightarrow \left\{\begin{array}{l} u_{yz}+y\,u_{xy}+u_x=0 \\
u_{xy}=u_{yz}=0\end{array}\right. \Longrightarrow u_x=0
$$

Hence, the initial system is equivalent to $\{ u_x=u_y=u_z=0 \}$.
\end{example}

In the case of a linear system (S) 
$$S_i=\sum_{j=1}^n D_{ij}\bullet u_j, \ i=1,\ldots m, \;\; D_{ij}\in A$$ 
we associate to (S) the submodule $P=P(S)\subset A^n$ generated by the
columns of a presentation matrix $D\in Mat(m,n;A)$, and, finally, 
a factor-module $M(S):=A^n/P(S)$. We can simplify the system 
finding a special presentation matrix, or we can read properties of 
the system from computable invariants of the module $M(S)$. 

In the example above, the system can be written as
$\begin{pmatrix} \d_z + y\d_x \\ \d_y \end{pmatrix} \bullet (u) = \begin{pmatrix} 0 \\ 0 \end{pmatrix}$. Hence, the system algebra is $A = \K(x,y,z)\langle \d_x,\d_y,\d_z \mid \d_x x = x \d_x +1, \d_y y = y \d_y +1, \d_z z = z \d_z +1 \rangle$ (that is the 3rd Weyl algebra) and the presentation matrix for the system module $M(S)$, written in columns of the original presentation (that is, transposed to the usual raw presentation) is $P(S) = (\d_z + y\d_x, \d_y) \in A^{1\times 2}$. As a submodule of $A$, $P(S)$ is an ideal and it has two polynomial generators $\{\d_z + y\d_x, \d_y\}$. The Gr\"obner basis of $P(S)$ is equal to $Q(S) = \langle \d_x, \d_y, \d_z \rangle$ 
and hence, $M(S) = A/Q(S) \cong \K(x,y,z)$ as $A$-module. Thus, $\dim_{\K(x,y,z)} M(S) = 1$ is the dimension of the solution space of $S$.

\subsection{Gr\"obner basis algorithm and elimination tools}

The notion of Gr\"obner basis can be given in a common way for
different classes of algebras. Recall the basic notation
for monomials and monomial ordering.
We shall use the short notation 
$\d^{\alpha}:=\d_1^{\alpha_1}\d_2^{\alpha_2}\ldots\d_n^{\alpha_n}$,
$\alpha \in \N^n$. 
Finitely generated operator algebras, which we are dealing with, 
have infinite dimension as $\K$-vector spaces. 
The infinite set of monomials constitutes this basis.

\begin{itemize}
\item For operator algebras in operators $\{\d_1,\ldots,\d_n\}$ with constant 
coefficients, the monomials are $\{\d^{\alpha} \mid \alpha \in \N^n\}$, 
they form the basis of an algebra over the field $\K$.

\item In the case where the coefficients are polynomial in $\{x_1,\ldots,x_m\}$, 
the monomials are \\
$\{x^{\alpha}\cdot\d^{\beta} \mid \alpha \in \N^m, \beta\in\N^n\}$ 
and they form the basis of an algebra over $\K$.

\item When the coefficients are rational functions, 
the monomials $\{\d^{\alpha} \mid \alpha \in \N^n\}$ 
constitute the basis of an algebra over $\K(x_1,\ldots,x_m)$.
\end{itemize}

We dealing not only with ideals of an algebra $A$, but also with submodules of a
free module $A^r = \oplus^{r}_{i=1} A e_i$, 
where $e_i$ is the canonical $i$--th basis vector. 
We extend the notion of a monomial to $A^r$ by supplying a monomial 
with one of the basis vectors. Clearly, if  $\Mon(A):=\{m_{\alpha} \}$ 
is the set of monomials of $A$, bijective to $\N^r$, then the monomials of $A^r$ are
$\{m_{\alpha} e_i \mid \alpha \in \N^n, 1\leq i \leq r\}$.

\begin{definition}
A (global) monomial ordering on an algebra $A$ as before is a total ordering $\prec$ on the set of monomials 
$\Mon(A)$ bounded from below and compatible with the multiplication, i.e. it fulfills
the the following conditions for all $\alpha, \beta, \gamma \in \N^n$:
\begin{itemize}
\item $1 \prec m_\alpha$,
\item $ m_\alpha \prec m_\beta \; \Rightarrow \; m_{\alpha}m_\gamma \prec m_{\beta}m_{\gamma}$.
\end{itemize}

Since $\prec$ is total, any nonzero polynomial $f\in A$ can be uniquely sorted according
to its monomials. The highest term (that is monomial times nonzero coefficient)
is called the leading monomial of $f$.
We say, that $m_\alpha \mid m_\beta$ ($m_\alpha$ divides $m_\beta$), if
$\forall 1\leq i \leq n$ $\alpha_i \leq \beta_i$. Note, that the
divisibility is a partial ordering.

\end{definition}
Any monomial ordering is extendable to a module monomial ordering 
in several ways. The most common ways are: either sorting module monomials first by the monomial ordering and then 
by the number of component, or first by component and then by the monomial ordering.
\begin{definition}
Given a monomial ordering $\prec$ of $A$, then monomial orderings $\prec_{top}$ (term-over-position) and $\prec_{pot}$  (position-over-term)
on the set of monomial of $A^r$ are defined by:
$$(m_\alpha,e_i)\prec_{top} (m_\beta,e_j)\ \mbox{iff} \
(m_\alpha \prec m_\beta \ \mbox{or if} \  m_\alpha = m_\beta, \mbox{then} \ i<j),
$$
respectively  
$$(m_\alpha,e_i)\ \prec_{pot}(m_\beta,e_j)\ \mbox{iff} \ (i<j
\ \mbox{or if} \ i=j, \mbox{then} \ m_\alpha \prec m_\beta).
$$ 
\end{definition}
\begin{definition}
A Gr\"obner basis of a submodule $M \subset A^r$ is a finite subset $G\subset M$ that satisfies the following property. 
For any $f\in M\setminus\{0\}$, there exists a element of the basis $g\in G$, such that the leading monomial $lm(f)=m_\alpha e_i$ is divisible by $lm(g)=m_\beta e_i$, i.~e. $\beta \mid \alpha$. 
\end{definition}

An immediate but very important application of Gr\"obner basis is
the {\em normal form} of a vector of polynomials. Namely, if $G=\{g_1,\ldots,g_m \}$ is a
Gr\"obner basis of a submodule $M\subset A^r$, then for any $v\in A^r$
there exist $w, a_i \in A$, such that
\[
v = \sum_{i=1}^m a_i g_i + w, \text{ where either } a_ig_i=0 \text{ or } \lm(v) \preceq \lm(a_i g_i) \text{ and either} w=0 \text{ or } \lm(w) \preceq \lm(v).
\]

One denotes $w = \NF(v, G)$ and calles $w$ to be a normal form of $v$ with respect to $G$. Note, that $v \in M = \langle G \rangle$ 
if and only if $w = \NF(v, G) = 0$.

A Gr\"obner basis $G = \{g_1,\ldots,g_m \}$ is called {\em reduced}, if 
for any $1\leq i\leq m$ and $j\not=i$, no monomial of $g_j$ is
divided by $\lm(g_i)$. Any Gr\"obner basis can be made reduced in
a finite number of steps. Notably, normalized reduced Gr\"obner
basis (that is, having $1$'s as leading coefficients of $g_i$'s)
is unique. It turns out, that normalized reduced normal form is unique, hence we address it by saying {\em the normal form}. \\



There are effective ways to compute a Gr\"obner basis, like the
{\em Buchberger's Algorithm}, {\em involutive algorithm} and {\em Faug{\`e}re's F4} or {\em F5 algorithm}. Gr\"obner bases have been implemented in all major computer algebra
systems. More details for the commutative case can be found meanwhile in any standard textbook on computer algebra, e.~g. in \cite{GPS}. Consult with the \cite{Mgfun,LVdiss} for the non-commutative case of operators with variable coefficients.



It is important to mention, that applied to a module, generated by
the columns of a constant matrix, the result of a 
Gr\"obner basis algorithm (with respect to position-over-term ordering) is identical to the result of Gaussian elimination. \\

Using the special monomial ordering, a Gr\"obner basis algorithm can be used to eliminate some of the variables  $\{ u_i \;|\; i\in I\}$  of a given system, i.e. to compute a basis of $M_I:=M\cap A_I^r$, $A=A_I\langle u_i,i\in I \rangle$.

\begin{lemma}
\label{ElimVar}
{\bf (Elimination of variables)}.
Let $\prec$ be an elimination monomial ordering for $\{ u_i \;|\; i\in I\}$  on $\Mon(A)$ (that is,  $m_\alpha\in A_I$ and $j\not\in I$ implies $m_\alpha\prec u_j$). Let $G$ be a Gr\"obner basis of $M$, then $G\cap A_I$ is a Gr\"obner basis of $M_I$.
\end{lemma}
Note, that a lexicographical ordering of monomials by $u_1>u_2>\ldots >u_n$
induces an elimination monomial ordering for any set $\{u_i,\ldots ,u_n\}$.

We can also eliminate module components, and usually it is much easier than the elimination of variables.

\begin{lemma}
\label{ElimComp}
{\bf (Elimination of components)}.
Let $G$ be a Gr\"obner basis of a submodule $M\subset A^r$ with respect to
the module monomial ordering $\prec_{pot}$, let 
$F_s:=Ae_1\oplus \dots \oplus Ae_s \subset A^r$
be the free submodule of the first $s$ components, then $G\cap F_s$ is a Gr\"obner basis of $M\cap F_s$.
\end{lemma}
The proofs of both lemmata on elimination are easy and can be found in e.~g. \cite{GPS, LVdiss} for various situations.

\begin{remark}
We want to stress the fact, that the algorithm we propose, using the operator formulation of the problem, will always be faster, than the algorithm in difference algebra, which is used by Gerdt et al. in \cite{BGM}. The difference in complexity lies in the number of variables/components one uses and the intrinsic differences between two similar-looking elimination concepts. 
Computing in the case, when the functions and discretizations of their derivatives $u,u_t,u_{xx}$ etc. appear as variables (difference algebra approach), one must distinguish between the two multiplications, namely the action one (denoted by $\bullet$), appearing when difference operators act on $u$'s, and the operator one (denoted by $\cdot$), used for composition of difference operators themself. 
One either uses the involutive basis approach with its division of all variables into multiplicative and non--multiplicative, or forbids the multiplications between $u$'s in the Gr\"obner basis algorithm. In addition, one has to employ a complicated elimination ordering, which respects the special role of $u$'s.

\indent
By passing to the submodule over a ring of difference operators, we do not use $u$'s at all.
We use the linearity of operators in order to present the equations as linear operators themselves, presented by polynomials in difference operators with constant coefficients, involving parameters. Thus, we use less variables, simple and efficient module component elimination orderings. Moreover, we shift the attention of  Gr\"obner basis algorithm from single polynomials rather to their components, which results in easier and faster computation, not speaking on much optimized usage of memory. 
\end{remark}

\section{Three Equivalent Approaches and the Main Theorem}
\label{s3}

Assume we are dealing with $m$ spatial variables $x_1,\ldots,x_m$ and one temporal variable $t=x_{m+1}$. We denote $x^\alpha := x_1^{\alpha_1} \cdots x_m^{\alpha_m} t^{\alpha_{m+1}}$
for $\alpha \in \N^{m+1}$ and $|\alpha | = \sum \alpha_i$. 
Then we use notations
\[
u_{x^\alpha} = u_{\alpha} := 
\frac{\partial^{|\alpha |}u}{\partial x^\alpha} = 
\frac{\partial^{|\alpha |}u}{\prod \partial x_i^{\alpha_i}  }
\]

A single linear PDE with constant coefficients can be written as follows
\begin{equation}
\label{EqCont}
P \equiv \quad \sum _{\beta \in B} c_{\beta} u_{x^\beta} = 0, 
\end{equation}
where $B \subset \N^{m+1}$ is a finite set and $c_{\beta} \in \K$.

In order to compute approximations, 
we rewrite the equation in an arbitrary inferior point 
$\bar{\iota} = (i_1,\ldots,i_m,n) \in \mathbb{G} := \Z \triangle x_1 \times \cdots \times \Z \triangle x_n \times \Z \triangle t$ (where $G$ is often identified with $\Z^{m+1}$) of the grid, that is

\begin{equation}
\label{EqDiscr}
\sum _{\beta \in B} c_{\beta} u^{\bar{\iota}}_{x^\beta} = 0.
\end{equation}

On the grid, one needs to give an approximation to a function
$u^{\bar{\iota}}_{x^\beta}$ through a linear combination of
similar expressions, coming from differentials of order, lower than $\beta$,
 that is $\forall \bar{\iota}\in \mathbb{G}, \forall \beta \in B\setminus\{0\}$

\begin{equation}
\label{AppDiscr}
A_\beta \equiv \quad
u^{\bar{\iota}}_{x^\beta} = \sum_{\gamma \in \Gamma, \bar{\jmath} \in \mathbb{G}} d_{\gamma,\bar{\jmath}} u^{\bar{\jmath}}_{x^\gamma} \text{ with } \mid \gamma \mid < \mid \beta \mid, 
d_{\gamma,\bar{\jmath}} \in \K.
\end{equation}
Note, that all but finite number of $d_{\gamma,\bar{\jmath}}$ are zero. \\

We say, that a general problem of approximation of a partial differential equation is \textbf{well-defined}, if for any $\beta \in B\setminus \{0\}$ 
(that is for any $u_{x^\beta}$, appearing in the equation with non-zero coefficient except for $u$ itself) there is a unique approximation from $A_\beta$.\\

Indeed, on the grid $\G$ we have natural shift operators $T_{x_i}: v^{\bar{\iota}} \mapsto v^{\bar{\iota} + \epsilon_i}$, where $\epsilon_i$ is the $i$-th canonical basis vector. That is $T_{x_i}(v(\ldots,\iota_i,\ldots)) = v(\ldots,\iota_i + \triangle x_i,\ldots))$ is a well-known forward shift operator. 
Clearly $T_{x_i}$ is invertible operator, since the inverse is just a backward shift operator. Thus, working with $T_{x_i}$ we allow exponents to be integers.
For an exponent vector $\alpha \in \Z^{m+1}$, denote 
$T^{\alpha} = T_{x_1}^{\alpha_1} \cdot \ldots \cdot T_{x_m}^{\alpha_m} T_{t}^{\alpha_{m+1}}$.
In what follows we will use the field of fractions $K(T):=K(T_{x_1},\ldots,T_{x_m},T_{t})$ of the polynomial ring $K[T]:=K[T_{x_1},\ldots,T_{x_m},T_{t}]$.

Thus, it is possible to use shift operators and rewrite the previous equation in a single generic point $\iota$ of the grid, as the following lemma shows.

\begin{lemma}
In the notations as before, there exist exponent vectors $\delta(\bar{\iota}), \delta(\bar{\jmath})$, $\jmath\in G$ such that
\begin{equation}
\label{AppDiscrShift}
T^{\delta(\bar{\iota})} (u_{x^\beta})^{\bar{\iota}} = 
\sum_{\gamma \in \Gamma, \bar{\jmath} \in \mathbb{G}} d_{\gamma,\bar{\jmath}} 
T^{\delta(\bar{\jmath})} (u_{x^\gamma} )^{\bar{\iota}}
\text{ or }
(u_{x^\beta})^{\bar{\iota}} = 
\sum_{\gamma \in \Gamma, \bar{\jmath} \in \mathbb{G}} d_{\gamma,\bar{\jmath}} 
T^{\delta(\bar{\jmath})-\delta(\bar{\iota})} (u_{x^\gamma} )^{\bar{\iota}}.
\end{equation}
\end{lemma}

\begin{proof}
For $1\leq k \leq m+1$ set $\kappa_k := \min \{ \jmath_k, \iota_k \mid \jmath\in\G, \gamma \in \Gamma, d_{\gamma,\bar{\jmath}}\not=0 \}$. Then, setting $\delta (\bar{\iota}) := \bar{\iota} - \bar{\kappa}$ and
respectively $\delta (\bar{\jmath}) := \bar{\jmath} - \bar{\kappa}$ we obtain two exponent vectors for the monomials in shift operators.
\end{proof}

According to the lemma, in what follows we will derive and encode approximations for 
functions on the grid with the help of shift operators. This allows us to drop the grid
point notation as soon as shift operators are present. In other words, we use shift operators to formulate the problem in a generic point of the grid.

There are several approaches to the computation of a finite difference scheme of a single partial differential equation.

\subsection{Mimicking Difference Algebra Approach}
\label{subsDA}

Consider the formal consequences of equalities $P$ as in (\ref{EqCont}) and $A_\beta$ as in (\ref{AppDiscr}) over the commutative ring $R_B:=K(T)[u_{x^{\beta}} \mid \beta \in B]$. Recall, that the variables $\{u_{x^{\beta}} \}$ are algebraically independent. In other words, we consider an ideal $I$ of the ring $R_B$, generated by $P \cup \{A_\beta \mid \beta \in B\}$. Since $B$ is finite, the ring $R_B$ is Noetherian and contains a subring $R:=K(T)[u]$. Hence the ideal
$J := I \cap K(T)[u]$ exists and it is computable (e.~g. by the elimination of all but one variables as in Lemma \ref{ElimVar}). Since $R$ is a 
principal ideal domain, $J$ is generated by a single element, say $p \in K(T)[u]$. Clearing denominators, we obtain a polynomial expression $\tilde{p} \in K[T][u]$. Dividing its leading coefficient out, we obtain the unique result.

Note, that we do not work with \textit{difference ideal} in a difference ring, but
just mimic a technique for presenting the objects.

\subsection{Algebraic Analysis Approach}
\label{subsAA}

Let us order the set $\{u_{x^{\beta}} \mid \beta \in B\}$ according to the monomial ordering (hence by the total degree as well), starting, say, from the highest appearing $\beta_{max}$, and write the resulting ordered list as a column vector $U$. Then it has the form $U =[u_{x^{\beta_{max}}},\ldots,u]^T$. Since $P$ and $A_\beta$ are linear equations with coefficients in $K[T]$ in the entries of $U$, we put each equation as a row in a matrix, say $M$, with entries in $K[T]$, such that $M \bullet U = 0$, where $\bullet$ stands for the action of shift operators with coefficients in $K$ on functions in discrete arguments. Then we can perform operations from the left on the matrix $M$ only, without engaging the $u$'s. This can be seen as the application of one of the principles of algebraic analysis. We compute then the intersection of $K[T]$-module $M$ with the free submodule, generated by the last component of $U$. The latter is an ideal $J\subset K[T]$ of all polynomials $p$ in shift operators, such that $p \bullet u = 0$. 

\subsection{Term Rewriting System Approach}
\label{subsRS}

Consider the equations from (\ref{AppDiscrShift}) in the monic form, that is
\[
u_{x^\beta} = 
\sum_{\gamma \in \Gamma, \bar{\jmath} \in \mathbb{G}} d_{\gamma,\bar{\jmath}} 
T^{\delta(\bar{\jmath})-\delta(\bar{\iota})} u_{x^\gamma}.
\]

Let us treat them as rewriting rules for symbols $\{ u_{x^\beta} \mid \beta \in B\setminus\{0\} \}$, which rewrites every appearance of 
$u_{x^\beta}$ with the sum on the right hand side. We call such system $S$.
Since the latter always involves $u_{x^\gamma}$ if and only if $\gamma \prec \beta$, we do the following. 
At first, we order occurring variables as in the approach \ref{subsAA}, that is descending with respect to the monomial ordering $\{u_{x^{\beta_{max}}},\ldots,u\}$. Then, in the same sequence, the variable $u_{x^{\varepsilon}}$ is substituted (replaced) with the right hand side of the corresponding approximation $A_\varepsilon$. According to the ordering, the result of the substitution does not contain variables, which are higher, according to the monomial ordering. Hence the result of the substitution is unique and consists of shift operators, acting on $u$ only. 

\begin{theorem}
\label{MainT}
For a single linear partial differential equation with constant coefficients,
the following three methods lead to the same result and hence they are equivalent:
\begin{enumerate}
\item As in \ref{subsDA}, $I \cap K(T)[u] = \langle f \rangle$, where $f$ is chosen to be monic polynomial from $K[T][u]$;
\item As in \ref{subsAA}, for $r=|B|$ holds $K[T]^r \supset M \cap K[T]e_r = J \subset K[T]$. Moreover, $J = \langle g \rangle$, where $g\in K[T]$ is chosen to be monic;
\item As in \ref{subsRS}, $P \ra_S h$ and $h$ is normalized in addition.
\end{enumerate}
\end{theorem}

\begin{proof}
a) $A_\beta$ is already a Gr\"obner basis in $K(T)[u_{x^\beta}]$
by the product criterion, because the leading monomials of its elements are coprime (since the set $\{u_{x^\beta} \mid \beta \in B\}$ is algebraically
independent. Since $\NF(P,A_\beta) \in K(T)[u]$, we obtain that $\{\NF(P,A_\beta)\} \cup A_\beta$ is a Gr\"obner basis of $P\cup A_\beta$. The uniqueness follows from the uniqueness of normalized reduced normal form \cite{GPS}. \\

b) Proceeding with the vector as above and starting with higher leading monomials, the matrix representation $M$ of the set $A_\beta$ is already 
upper triangular (with its entries in $K[T]$). Moreover, the last row has 
at exactly two nonzero elements (say, the last two ones in that row). 
Thus $M$ is already in a row-reduced form. Note, that making complete reduction
of the rows will produce a matrix $M'$, where each row contains exactly
two nonzero elements: 
$(0,\ldots,0,f_i(T),0,\ldots,0,f_r(T))$ with $f_i,f_r \in K[T]$.
Hence $M'$ simplifies the set of approximations and correspond to
the completely reduced Gr\"obner basis. 

Now, we append to $M$ (or $M'$, what is equivalent) the row $P'$, corresponding to the discretized equation $P$. The computation of the Gaussian elimination of the resulting matrix amounts in reductions of $P'$ with the rows of row-reduced $M'$. The result of such reductions is a row vector with the only nonzero entry $f(T)\in K[T]$ at it last position. Since it is the same as the reduced normal form $\NF(P',M')$ and $M'$ is a Gr\"obner basis, the ideal $J$ of the Theorem is principal ideal, generated by $f(T)$. Then the normalized $f(T)$ is single operator, acting on $u$, that is the difference scheme. \\

c) Note, that the application of rewriting rules in the sequence, as described above, leads to the $\NF(P,A_\beta)$. Since by a) $A_\beta$ is a Gr\"obner basis with respect to any monomial well-ordering, this normal form is unique.
Fixing a monomial ordering, we can produce another rewriting system $S'$ by applying rewriting rules to every right hand side of $S$, starting ascendingly from the smallest nonzero $\beta \in B$. Then $S'$ becomes $\{ (u_{x^{\beta}} \to f_\beta(T) u) \mid \beta\in B\setminus\{0\}\}$, where $f_\beta(T)\in K[T^{\pm 1}]\subset K(T)$. It is straightforward, that $S'$ does not depend on the sequence of reductions like $S$ anymore and $S'$ is confluent. The reduction of $P$ with respect to $S'$ is the same as with respect to $S$, hence it is unique.
\end{proof}

\begin{remark}
\label{why3}
Note, that the equivalences of the previous Theorem do not hold in general.
Namely, Algebraic Analysis Approach \ref{subsAA} works only for linear
PDE, since it relies on the module structure, which is linear per definition.
Moreover, both a) and b) do not necessarily work in the case of
variable coefficients. In the latter, there are different concepts
for discretization. As soon as one deals with algebras, where $x$ and $T_x$ do not commute anymore, the left normal form of a vector (the computation uses subtractions of left multiples of an approximation and thus invokes non-commutative multiplication) 
is not necessarily the result of rewriting of any term $u_{x^\beta}$ (which just plugs the right hand side expression into the place where the term resides, just not
invoking non-commutative multiplication). \\

Indeed the only method, which works for all cases of systems, is the
Term Rewriting System Approach \ref{subsRS} ($A_\beta$ is a Gr\"obner
basis for arbitrary linear approximations).
\end{remark}


\section{Generation of difference schemes}


Armed with the methods from the previous section, we proceed with
the generation of schemes for linear equation with constant coefficients.
We prefer the method of algebraic analysis from \ref{subsAA}, in contrast to Gerdt et al. \cite{BGM}, who used the method of difference algebra because
of the reasons of practical complexity, which is significantly lower
in the approach \ref{subsAA}. However, Gerdt et al. systematically follow
the difference algebra approach for nonlinear equations, where in \cite{BGM} they have obtained an interesting nice-behaving scheme of the order 3 for the the original equation of order 2, which however does not contain switches as usual schemes.

A big class of PDE's from mathematical physics are 
of conservative type in space-time with coordinates $(t,x)\in \R^4$:
$$\partial P/\partial t - \partial Q/\partial x=0.$$
By Gauss formulae this equation can be transformed into its conservation law form:
$$ \int_{\partial V}(Qdt+Pdx)=0.$$
We choose some discretized integration contours and approximations rules for 
the integrals and proceed as above. The difference schemes, which we obtain by 
elimination are by construction fully consistent \cite{BGM} .

\subsection{Approximation Rules and their Operator Form} 

The most general way to approximate a PDE is to use the integral
 relations like $\int^{t_{n+1}}_{t_n} u_{tt}(x,t) dt = u_t(x,t_{n+1}) -u_t(x,t_n)$ together with further approximations of differentials and integrals. \\

However, a large class of equations might be written in a so--called \textit{conservation law form}, which can be obtained e.g. by applying the Green's formula. For example, the equation $\frac{\d Q}{\d x} - \frac{\d P}{\d y} = 0$ is equivalent to the equation $\oint\limits_{\Gamma} P dx + Q dy = 0$ for arbitrary piecewise smooth closed contour $\Gamma$.

\textbf{Contour approximations.} There are many possibilities for choosing contours and
approximations. Since we are using rather rectangular than quadratic grids, the two most used approximation on contours are the node points of the rectangle and the middle points of
a grid with a double distance, as illustrated by the pictures below.

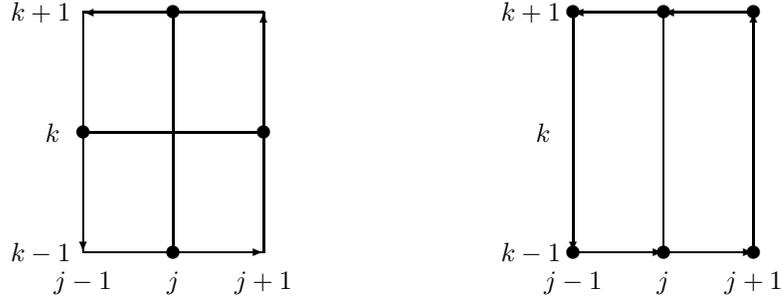
\begin{figure}[htbp]
\unitlength=0.8mm
\special{em:linewidth 0.4pt}
\linethickness{0.5pt}
\begin{picture}(80.00,61.00)(-18, 0)
\put(21.00,10.00){\vector(1,0){30.00}}
\put(51.00,10.00){\vector(0,1){40.00}}
\put(51.00,50.00){\vector(-1,0){30.00}}
\put(21.00,50.00){\vector(0,-1){40.00}}
\put(21.00,30.00){\line(1,0){30.00}}
\put(36.00,10.00){\line(0,1){40.00}}
\put(36.00,50.00){\circle*{2.00}}
\put(51.00,30.00){\circle*{2.00}}
\put(21.00,30.00){\circle*{2.00}}
\put(36.00,10.00){\circle*{2.00}}
\put(14.00,10.00){\makebox(0,0)[cc]{$k-1$}}
\put(16.00,30.00){\makebox(0,0)[cc]{$k$}}
\put(14.00,50.00){\makebox(0,0)[cc]{$k+1$}}
\put(21.00,5.00){\makebox(0,0)[cc]{$j-1$}}
\put(36.00,5.00){\makebox(0,0)[cc]{$j$}}
\put(51.00,5.00){\makebox(0,0)[cc]{$j+1$}}
\end{picture}
\begin{picture}(80.00,61.00)(-18, 0)
\put(21.00,10.00){\vector(1,0){15.00}}
\put(36.00,10.00){\vector(1,0){15.00}}
\put(51.00,10.00){\vector(0,1){40.00}}
\put(51.00,50.00){\vector(-1,0){15.00}}
\put(36.00,50.00){\vector(-1,0){15.00}}
\put(21.00,50.00){\vector(0,-1){40.00}}
\put(36.00,10.00){\line(0,1){40.00}}
\put(36.00,50.00){\circle*{2.00}}
\put(36.00,10.00){\circle*{2.00}}

\put(21.00,10.00){\circle*{2.00}}
\put(51.00,10.00){\circle*{2.00}}
\put(51.00,50.00){\circle*{2.00}}
\put(21.00,50.00){\circle*{2.00}}

\put(14.00,10.00){\makebox(0,0)[cc]{$k-1$}}
\put(16.00,30.00){\makebox(0,0)[cc]{$k$}}
\put(14.00,50.00){\makebox(0,0)[cc]{$k+1$}}
\put(21.00,5.00){\makebox(0,0)[cc]{$j-1$}}
\put(36.00,5.00){\makebox(0,0)[cc]{$j$}}
\put(51.00,5.00){\makebox(0,0)[cc]{$j+1$}}
\end{picture}
\end{figure}

Although by applying the Green's formula we lower the order
of an equation by 1, the approximation formulas, derived from the contour are usually more complicated, than the approximations, derived from the
original equation and integral relations. Note, that this is not a problem for an implementation, since complicated manipulations with polynomial expressions can be performed effectively with modern computer algebra systems.

\textbf{Approximation of differentials via Taylor series.} Applying the Taylor expansion up to the 2nd order, we obtain $u(x \pm \tri x) = u(x) \pm \tri x u_x(x) + \frac{\tri x^2}{2} u_{xx}(x) + \mathcal{O}(\tri x^3)$. Hence,

$u_x(x) = \dfrac{u(x+\tri x) - u(x)}{\tri x} + \mathcal{O}(\tri x)$ (forward difference)

or $u_x(x) = \dfrac{u(x)-u(x-\tri x)}{\tri x} + \mathcal{O}(\tri x)$ (backward difference).

Subtracting these two equalities we get 
$u(x+\tri x) - u(x) + u(x)-u(x-\tri x) = 2\tri x u_x(x) + \mathcal{O}(\tri x^3)$,
hence $u_x(x) = \frac{u(x+\tri x) - u(x-\tri x)}{2 \tri x} + \mathcal{O}(\tri x^2)$ (\textit{central $1^{st}$ order difference}).

Adding these two equalities and rewriting the result, we get the \textit{central $2^{nd}$ order difference}

$\dfrac{u(x+\tri x) - 2 u(x) + u(x-\tri x)}{\tri x^2} = u_{xx}(x) + \mathcal{O}(\tri x^2 )$.\\

\textbf{Approximation of integrals.}
Closed Newton-Cotes formulas give rise to e.g. trapezoid an pyramid rules, whereas open Newton-Cotes formulas lead us to e.g. midpoint rule.\\

The trapezoid rule is expressed as follows:

\[
\int \limits_{x_0}^{x_0 + \tri x} f(x) dx =  \frac{1}{2} \tri x (f(x_0)+f(x_0 + \tri x))-\frac{1}{12 \tri x^3}f^{''}(\xi), \: x_0 \leq \xi \leq x_0 + \tri x.
\]

Hence, in the approximation we obtain, for $f(x)=u_x(x)$:
\[
u(x_0 + \tri x) - u(x_0) = \int \limits_{x_0}^{x_0 + \tri x} u_x(x) dx = 
\frac{1}{2} \tri x (u_x(x_0)+u_x(x_0 + \tri x)),
\]

and hence $(T_x -1)\bullet u = \frac{1}{2} \tri x (T_x+1) \bullet u_x$.\\

Pyramid (or Simpson's) rule looks as follows:
\[
\int \limits_{x_0}^{x_0 + 2\tri x} f(x) dx  = \frac{1}{3} \tri x (f(x_0)+4f(x_0 + \tri x)+f(x_0 + 2\tri x))-\frac{1}{90 \tri x^5} f^{(4)}(\xi),
\]

hence its difference form is $\frac{1}{3} \triangle x \cdot (T^2_x+4T_x+1) \bullet u_x = (T^2_x-1) \bullet u$. \\

Open Newton-Cotes formula for 1 point,
\[
\int \limits_{x_0}^{x_0 + 2\tri x} f(x)dx = 2\tri x f(x_0 + \tri x) + \mathcal{O} (\tri x^2 f^{'}), 
\]
 leads us to the midpoint formula $\triangle x \cdot T_x \bullet u_x = (T^2_x-1) \bullet u$. \\

\textbf{Summary.} In a form of difference operators, we gather the most used approximations 
\begin{itemize}
\item \textbf{Forward difference} $(\triangle x, \; 1-T_x) \bullet (u_x, u)^T = 0$
\item  \textbf{Backward difference} $(\triangle x \cdot T_x, \; 1-T_x) \bullet (u_x, u)^T = 0$
\item  A $1^{st}$ \textbf{order central appr.} $(2\triangle x \cdot T_x, \; 1-T_x^2) \bullet (u_x, u)^T = 0$
\item A $2^{nd}$ \textbf{order central appr.} $(- \triangle x^2 \cdot T_x, \; (1-T_x)^2) \bullet (u_{xx}, u)^T = 0$ 
\item \textbf{Trapezoid rule} $(\frac{1}{2} \triangle x \cdot (T_x+1), 1-T_x)\bullet (u_x,u)^T = 0$
\item \textbf{Midpoint rule} $(2 \triangle x \cdot T_x, 1 - T^2_x) \bullet (u_x,u)^T = 0$.
\item \textbf{Pyramid rule} $(\frac{1}{3} \triangle x \cdot (T^2_x+4T_x+1), 1-T^2_x) \bullet (u_x,u)^T = 0 $
\item \textbf{Lax method}\footnote{used in e.g. discretization of the advection equation} $(2 \tri t \cdot T_x, T^2_x - 2T_tT_x +1)\bullet (u_t, u)^T = 0$
\item \textbf{Parametric temporal difference} for $0\leq \theta \leq 1$: 
$(\tri t \cdot(\theta T_t +(1-\theta)), 1- T_t) \cdot (u_t, u)^T = 0$. If $\theta=0$ resp. $\theta=1$, it becomes forward resp. backward difference.
\end{itemize}

Assume that the difference scheme involves quantities $\tri x_1,\ldots,\tri x_m, \tri t$ and originates from a typical set of approximations. 
By definition, difference scheme is of the smallest difference order, hence
the shift polynomial $p$, describing it, is irreducible.
It turns out, that in many situations we want to present $p$ as the sum
of products of operators. Taking into account the future applications
to the von Neumann stability, we propose the following.

\begin{definition}
A \textbf{semi-factorized} presentation of a linear difference scheme of order $O(\tri x_1^{b_1}, \ldots, \tri x_m^{b_m},\tri t^c)$ to be the sum 
$p = \tri x_1^{b_1} p_1 + \ldots + \tri x_m^{b_m} p_m + \tri t^c p_t$ for $p_i \in K[T]$, such that $p_i$ does not involve $\tri x_i$ in its coefficients and
most (if not all) $p_j$ do not involve $\tri x_1,\ldots,x_m, \tri t$.
\end{definition}

Unlike nodal form, a semi-factorized form allows compact descriptions of very complicated and higher dimensional schemes. Note, that in the examples we present, it turns out that there exists a unique (up to constant factors) semi-factorized presentation. In the implementation we have a method for computing a
semi-factorized form constructively.

\begin{example}
\label{Heat}
Consider the 1D heat equation $u_t - a^2 u_{xx} = 0$ for a parameter $a$.\\

We approximate $u_t$ with backwards difference 
$\triangle t \cdot T_t \bullet u_t = (T_t-1)\bullet u$,
or, in the nodes of the grid,
$\triangle t \cdot (u_t)^{m+1}_i = (u)^{m+1}_i - (u)^{m}_i$. \\

As for $u_{xx}$, it is approximated with the 2nd order weighted centered space method, that is
\[
\triangle x^2 \cdot T_x \bullet u_{xx} = (\theta T_t + (1-\theta))\cdot (T_x - 1)^2 \bullet u,
\text{ where } 0\leq \theta \leq 1.
\]

In such a way, we obtain a matrix formulation of the problem
\[
\begin{pmatrix}
  1 & -a^2 & 0 \\
  - \triangle t \cdot T_t  & 0 & T_t-1 \\
  0 & -\triangle x^2 T_x T_t & (\theta T_t + (1-\theta)) \cdot (Tx-1)^2
\end{pmatrix}
\bullet
\begin{pmatrix}
u_t \\ 
u_{xx}\\ 
u
\end{pmatrix}
=0.
\]

By computing a Gr\"obner basis (with the algebraic analysis approach), 
we obtain a single polynomial in shift operators for the scheme
\[
-a^{2} \tri t \theta T_x^{2}T_t+ a^{2} \tri t (\theta-1) T_x^{2}+(2a^{2} \tri t \theta+\tri x^{2}) T_x T_t-(2a^{2} \tri t (\theta-1) + \tri x^{2})T_x-a^{2} \tri t \theta T_t+a^{2} \tri t( \theta-1)
\]

Its semi-factorized form is $\tri x^2 T_x (T_t-1)  - a^2 \tri t (T_x-1)^2 (\theta T_t+ 1-\theta) = 0$. In the following example we show \textsc{Singular} code for obtaining these objects and for producing a nodal presentation of the scheme, which is
\[
\frac{1}{a^{2} \tri t}\cdot (u^{n+1}_{j+1}-u^{n}_{j+1}) - \frac{ \theta}{ \tri x^{2}}\cdot (u^{n+1}_{j+2} -2 u^{n+1}_{j+1}+u^{n+1}_{j}) - \frac{(1-\theta)}{ \tri x ^{2}}(u^{n}_{j+2}-2 u^{n}_{j+1}+u^{n}_{j}) = 0.
\]

The scheme we obtained is called FTCS if $\theta=0$, BTCS if $\theta=1$ and Crank-Nicholson, if $\theta=\frac{1}{2}$. \\

If is easy to see, that this scheme is consistent with the original differential equation for any $\theta\in \R$. Namely, since $\frac{u^{n+1}_{j+1}-u^{n}_{j+1}}{\tri t} = u_t + \mathcal{O}(\tri t)$ and $\frac{u^{n+1}_{j+2} -2 u^{n+1}_{j+1}+u^{n+1}_{j}}{ \tri x^{2}} = u_{xx} + \mathcal{O}(\tri x^2)$, we have 
\[
\frac{1}{a^{2}} u_t - \theta  u_{xx} - (1-\theta)  u_{xx} = \frac{1}{a^{2}} u_t - u_{xx} = \mathcal{O}(\tri t) + \mathcal{O}(\tri x^2).
\]

Thus, the order of the scheme is $(\tri t, \tri x^2)$.
\end{example}

\begin{example}
\label{HeatSing}

In this example we do computations with \textsc{Singular} and \texttt{findifs.lib}. As we see from the matrix formulation, our \textit{parameters} are $ \tri t, \tri x, a, \theta$, to which we add the parameter $d$, which will be needed later for the check of stability.
The \textit{variables} of our ring are $T_t$ and $T_x$. 
We define the ring in \textsc{Singular} and the matrix of equations as follows:
\begin{verbatim}
ring r = (0,a,dx,dt,theta,d),(Tx,Tt),(c,Dp); 
matrix M[3][3]=
1, -a^2, 0,                                   // the equation itself
-dt*Tt, 0, Tt-1,                              // appr. u_t with backward difference
0, -dx^2*Tt*Tx,(theta*Tt+(1-theta))*(Tx-1)^2; // appr. u_xx with theta-centered space 
\end{verbatim}

Now we transpose the module and then call the \texttt{std} routine for getting the Gr\"obner basis.

\begin{verbatim}
module R = module(transpose(M)); module S = std(R);
print(S);
==> 0,     0,           1,     
    0,     (-a^2*dt)*Tt,(-a^2),
    S[3,1],Tt-1,        0  
\end{verbatim}

As we can see, the first column vector is the only one with the values in the 3rd component only. \texttt{S[3,1]} is displayed since this polynomial (which describes the difference scheme) is big.

\begin{verbatim}
poly p = S[3,1]; p; // assign and print the answer
==>(-a^2*dt*theta)*Tx^2*Tt+(a^2*dt*theta-a^2*dt)*Tx^2+(2*a^2*dt*theta+dx^2)*Tx*Tt+
   (-2*a^2*dt*theta+2*a^2*dt-dx^2)*Tx+(-a^2*dt*theta)*Tt+(a^2*dt*theta-a^2*dt)
\end{verbatim}
We proceed with the semi-factorized form and visualization.
\begin{verbatim}
LIB "findifs.lib"; // load the library for schemes
ideal I = decoef(p,dt); // auxiliary procedure
I;  // the sum of elements of I gives p
==>I[1]=(dx^2)*Tx*Tt+(-dx^2)*Tx
   I[2]=(-a^2*dt*theta)*Tx^2*Tt+(a^2*dt*theta-a^2*dt)*Tx^2+(2*a^2*dt*theta)*Tx*Tt+
        (-2*a^2*dt*theta+2*a^2*dt)*Tx+(-a^2*dt*theta)*Tt+(a^2*dt*theta-a^2*dt)
\end{verbatim}

From this structure, we can obtain the semi-factorized operator form of the scheme:

\begin{verbatim}
factorize(I[1]); // we suppress the output
factorize(I[2]); // factors with multiplicities
==>[1]:
     _[1]=(-a^2*dt)
     _[2]=Tx-1
     _[3]=(theta)*Tt+(-theta+1)
  [2]:
     1,2,1
\end{verbatim}

Hence, the semi-factorized form is $\tri x^2 T_x (T_t-1)  - a^2 \tri t (T_x-1)^2 (\theta T_t+ 1-\theta) = 0$.

\begin{verbatim}
list L; L[1] = theta;
difpoly2tex(I,L); // we show though only a part of this string
==>\frac{-1}{a^{2} \tri t}\cdot (u^{n+1}_{j+1}-u^{n}_{j+1})+ ...
\end{verbatim}


The string above (in tex format) is the nodal presentation of the scheme,
which has already been demonstrated in the previous example.


\end{example}


\section{Symbolic methods for von Neumann stability analysis}
\label{sStability}

\subsection{Stability rings, morphisms and polynomials}
We refer the reader to e.~g. \cite{T} for details about stability.
Suppose that $t$ is the temporal variable and $x_1,\ldots,x_m$ are
the spatial variables. We start with a finite difference scheme, written in 
the nodal form on the uniform orthogonal grid with steps $\tri t, \tri x_1, \ldots,
\tri x_m$. We suppose to work in the interior region,  which is bounded, say,
by $L_1,\ldots,L_m$ in spatial directions. \\

In the von Neumann stability analysis, one presents the functions on the grid as discrete Fourier modes, that is
\[
\chi(u^n_{j_{1} j_{2} \dots j_{m}}) = g^n \prod^{m}_{k=1} e^{i j_k l_k \pi \tri x_{k}},
\]

where $\chi$ is a linear map, $g$ is a new symbolic variable, $0\leq \l_k \tri x_{k} \leq L_k$. We abbreviate $\beta_{j_k}:= \pi l_k \tri x_{k}$. Then, one substitutes this presentation of
nodes into the equation, performs simplifications and obtains a polynomial $G$ in one variable $g$ with constant coefficients. \\

The \textbf{von Neumann stability criterion} (as in \cite{T}) states, that if for every root $\xi$ of $G$, one has $|\xi| \leq 1$, then the difference scheme is stable.\\

The \textbf{Lax equivalence theorem} can be stated in the following form (adopted from \cite{T}). A consistent scheme for a well-posed linear initial value problem is convergent if and only if it is stable. For a well-posed linear initial-boundary-value problem, however, stability is only a necessary condition for convergence in general. \\

We do not address algorithms for algorithmic check of consistency of a difference scheme with its differential equation in this paper. There are several methods for doing it using algebraic tools like \cite{GV96, GR10}. However, for several equations we treat we show the usage of semi-factorized form of a scheme to
the positive conclusion about such consistency.

Let $A$ be the algebra of functions on a given grid. It is naturally a module over the
algebra $R$ of linear partial difference operators with constant coefficients $C[T_t,T_{x_1},\ldots,T_{x_m}]$ over some field $C\supseteq \Q(\tri t, \tri x_1,\ldots, \tri x_m)$.
The action of $R$ on discrete Fourier nodes, using the map $\chi$, can be written
as follows: 
\[
\chi(T^a_t \bullet u^n_{j_{1} j_{2} \dots j_{m}}) = g^a \cdot \chi(u^n_{j_{1} j_{2} \dots j_{m}}) \text{ and } \chi(T^b_{j_s} \bullet u^n_{j_{1} j_{2} \dots j_{m}}) = e^{i j_{s} l_{s} \pi \tri x_{k}} \cdot \chi(u^n_{j_{1} j_{2} \dots j_{m}}) \text{ for all } j_k.
\] 
The map $\chi$ and this action give rise to the constructive homomorphism of $C=\Q(\tri t, \tri x_1,\ldots, \tri x_m)$-algebras
\[
\chi: C[T_t,T_{x_1},\ldots,T_{x_m}] \longrightarrow C\bigl ([i,sin_{x_1},cos_{x_1},\ldots,sin_{x_m},cos_{x_m}]/J_m \bigr ) [g],
\]

where $J_m = \langle i^2+1,sin^2_{x_1}+cos^2_{x_1}-1,\ldots,sin^2_{x_m}+cos^2_{x_m}-1\rangle$ is the ideal. We denote this \textbf{constructive stability morphism} by the same letter $\chi$ and note its $C$--linearity. It is defined by
its values on the generators of the source algebra and is given by 
$\chi(T_t) = g$  and $\chi(T_{j_s}) = e^{i l_{s} \pi \tri x_{s}} = $
$\cos{\beta_s} + i \cdot \sin{\beta_s}$ for all $1\leq s \leq m$.\\

The constructiveness of this approach and hence, its applicability in computer algebra systems, lies in the following. We choose the basic numeric field to be exact complex-rational numbers 
$\Q[i]/\langle i^2+1 \rangle$. On demand, we can do further algebraic extensions of this field.
Moreover, we avoid complex exponentials by passing to the sine and cosine, incorporating their natural algebraic relations in the factor ideal. Then, a stability morphism can be implemented and used 
in every computer algebra system, being able to compute Gr\"obner bases.\\

Now, let $P = \sum_{a,\alpha} c_{a,\alpha} T^a_t T^{\alpha}_{x}$ be the operator form of
the finite difference scheme $P \bullet u = 0$, where $T^{\alpha}_x$ stands for 
$T^{\alpha_1}_{x_1} \cdot T^{\alpha_m}_{x_m}$ for a multi-index $\alpha=(\alpha_1,\ldots,\alpha_m) \in \N^m$. Then, 
\[
\chi(P) = \sum_{a,\alpha} c_{a,\alpha} \chi(T_t)^a \chi(T_{x})^{\alpha} = \sum_{a,\alpha} c_{a,\alpha} \prod^m_{k=1}(\cos{\beta_k} + i \cdot \sin{\beta_k})^{\alpha_k} g^a = \sum_a d_a g^a
\]

is the univariate polynomial in $g$, which we call the \textbf{stability polynomial} of a given difference scheme. Obviously, the degree of $\chi(P)$ is the same as the highest degree of $T_t$ in $P$. \\

\begin{example}
\label{HeatST}

Let us continue with the example \ref{Heat}. In order to prepare the scheme for stability analysis, one can rewrite it as 

\[
u^{n+1}_{j+1}-u^{n}_{j+1} = a^2 d \bigl ( \theta\cdot (u^{n+1}_{j+2} -2 u^{n+1}_{j+1}+u^{n+1}_{j}) + (1-\theta)\cdot(u^{n}_{j+2}-2 u^{n}_{j+1}+u^{n}_{j}) \bigr),
\]

with $d:= \dfrac{\tri t}{\tri x^2}$. 
We prefer to work with the semi-factorized operator form of the scheme $\tri x^2 T_x (T_t-1)  - a^2 \tri t (T_x-1)^2 (\theta T_t+ 1-\theta) = 0$

Creating the stability ring and performing simplification and
factorization in it (see next example for the \textsc{Singular} code),
we obtain the following linear polynomial in the variable $g$.

\[
(i \cos+ \sin) \cdot (((-2a^2 d \theta) \sin + 2 a^2 d \theta+1)\cdot g+
(2 a^2 d \theta-2 a^2 d) \sin- 2 a^2 d \theta+ 2 a^2 d-1)
\]

The first factor $i\cdot \cos(\beta)+\sin(\beta) = e^{i \cdot \beta}$ is
ignored in stability analysis, since it is of magnitude 1.

\end{example}

\begin{example}
\label{HeatSTSing}

We continue with the Example \ref{HeatSing}. Define the
semi-factorized scheme again. 
\begin{verbatim}
poly P  = Tx*(Tt-1)  +  (-a^2)*d*(Tx-1)^2*((theta)*Tt+(-theta+1));
ring r2 = (0,a,theta,d),(Tx,Tt),(c,Dp);
poly P  = imap(r,P);
\end{verbatim}

Now, we create the stability ring \texttt{ST} (which will be
$\Q(a,\theta,d)[g,i,sin,cos]$) and a map $\chi$ from \texttt{r2} (it is $\Q(a,\theta,d)[T_x,T_t]$ to it: 

\begin{verbatim}
ring ST     = (0,a,d,theta),(g,i,sin,cos),lp;
ideal Rels  = std(ideal(i2+1,sin^2+cos^2-1)); // the ideal of relations
map chi     = r2,ideal(sin+i*cos,g);
poly P = chi(P); // the mapping
P = NF(P,Rels); P;  // reduction modulo the relations in Rel

==>(-2*a^2*d*theta)*g*i*sin*cos+(2*a^2*d*theta+1)*g*i*cos+ ...
ideal FP = factorize(P); // factorization
\end{verbatim}



The polynomial together with its factorization have been presented in the previous example.

\end{example}

So, we came from a system of linear equations to a single univariate
polynomial in the stability ring. The next problem we face is the following: \\
\begin{quote}
{\em
Given a univariate parametric polynomial $P$, find out, under which conditions on parameters all the roots of $P$ lie in the complex unit circle.
}
\end{quote}

As it have been already mentioned in \cite{HLS, FIDE}, this problem can be solved algorithmically with the help of the CAD (Cylindrical Algebraic Decomposition).

\subsection{Cylindrical Algebraic Decomposition}

The algorithm for CAD goes back to G.~Collins et al. It is one of the most important algorithms, used for quantifier elimination not only in real algebraic geometry \cite{BPR}. On the other hand, its algorithmic complexity is high and can be double exponential in the number of variables. Nevertheless, the universality of the method makes it very powerful and applicable to various problems.\\

A finite set of polynomials $\{p_1,\ldots,p_m\} \in \R[x_1,\ldots,x_n]$ induces a
decomposition (partition) of $\R^n$ into maximal sign--invariant cells. 
A \textbf{cell} in the algebraic decomposition of $\{p_1,\ldots,p_m\} \in \R[x_1,\ldots,x_n]$ 
is a maximal connected subset of $\R^n$ on which all the $p_i$ are sign invariant.

\begin{definition}
For $n\in \N$, let $\pi_n : \R^n \to \R^{n-1}$, $(x_1,\ldots, x_{n-1}, x_n) \mapsto (x_1,\ldots,x_{n-1})$ denote the canonical projection.
Let $\{p_1,\ldots,p_m\} \in \Q[x_1,\ldots,x_n]$.
The algebraic decomposition of $\{p_1,\ldots,p_m\}$ is called \textbf{cylindrical}, if
\begin{itemize}
\item For any two cells $C,D$ of the decomposition, the
images $\pi(C),\pi(D)$ are either identical or disjoint.
\item The algebraic decomposition of $\{p_1,\ldots,p_m\} \cap \Q[x_1,\ldots,x_{n-1}]$
is cylindrical.
\end{itemize}
\end{definition}

For instance, any algebraic decomposition of $\R^1$ is cylindrical.\\

There are several sophisticated implementations of the CAD algorithm. We are
using the one from the system \textsc{Mathematica}, where two commands,
\texttt{CylindricalDecomposition} and \texttt{Reduce} are available in the context of CAD. There are also freely available systems \textsc{QEPCAD} by C.~Brown \cite{QEPCAD} and \textsc{REDLOG} by A.~Dolzmann et al. \cite{REDLOG}.

\subsection{CAD and Stability}

\begin{example}
\label{HeatST2}

Let us continue with the examples \ref{Heat}, \ref{HeatST}. 
Let us represent the root of a stability polynomial as $\tfrac{c}{d'}$,
where $c= 2 a^2 d (1-\theta) \sin - 2 a^2 d (1 - \theta)+1$, $d'= 2 a^2 d \theta(1-\sin)+1$

Since $d' \geq 0$, we have to solve the inequality $-d \leq c \leq d$, that is $c+d \geq 0$ and $d\geq c$.
The first inequality $2 a^2 d (1-\sin) \geq 0$ is always satisfied, and the second  is equivalent to $a^2 d (2\theta-1)(1-\sin)+1 \geq 0$.
In this example, we compare the functions \texttt{CylindricalDecomposition} and \texttt{Reduce} of \textsc{Mathematica}

\begin{verbatim}
CylindricalDecomposition[{a^2*d*(2*theta-1)*(1-s) + 1 >= 0, 
-1 <= s <= 1, a > 0, d > 0}, {theta, a, d, s}]
\end{verbatim}

returns 
\begin{footnotesize}
\[
\bigl (\theta < \dfrac{1}{2} \: \&\& \: a > 0 \: \&\& \: (0 < d \leq - \dfrac{1}{-2a^2 + 4a^{2} \theta} ) \:\&\&\: -1 \leq s \leq 1 \bigr) \:\: || 
\]
\[
\bigl(d > -\dfrac{1}{-2 a^2 + 4 a^2 \theta} \;\:\&\&\:\; \dfrac{1 - a^2 d + 2 a^2 d \theta}{-a^2 d + 2 a^2 d \theta} \leq s \leq 1\bigr)  \:\: || 
\: \bigl(\theta> \dfrac{1}{2} \; \&\& \;  a > 0 \;\&\&\; d > 0 \;\&\&\;  -1\leq s \leq 1 \bigr).
\]
\end{footnotesize}
Whereas executing more specialized 
\begin{verbatim}
Reduce[a > 0 && d > 0 && 0 <= theta <= 1 && 
ForAll[s, -1 <= s <= 1, a^2*d*(2*theta - 1)*(1-s) + 1 >= 0], {theta, d}]
\end{verbatim}

gives us more informative and structured answer
\[
a > 0 \;\&\&\; ( 0 \leq \theta < \dfrac{1}{2} \;\&\&\; 0 < d \leq - \dfrac{1}{-2a^2 + 4a^{2} \theta} ) \;||\; (\dfrac{1}{2} \leq \theta \leq 1 \;\&\&\; d > 0),
\]

from which we conclude, that
\begin{itemize}
\item if $\frac{1}{2} \leq \theta \leq 1$, the scheme is unconditionally stable
\item if $0\leq \theta < \dfrac{1}{2}$, the scheme is stable under the condition 
$d=\frac{\tri t}{\tri x^2} \leq \dfrac{1}{2a^{2} (1-2\theta)}$.
\end{itemize}

The quantity $d=\frac{\tri t}{\tri x^2}$ if often called Courant (or Courant-Friedrichs-Lewy) number. It is classical to express conditions on the von Neumann stability in terms of the Courant number.
\end{example}

\begin{example}
Consider the 1D \textbf{advection equation} $u_t + a u_x = 0$. We approximate
$u_t$ with the parametric temporal method and $u_x$ with the trapezoid rule.
As a result, we obtain the difference scheme in the semi-factorized form
$\tri x \cdot (T_x+1)\cdot (T_t-1) + 2a \tri t \cdot (T_x-1)\cdot (\theta T_t-(\theta-1)) = 0$, which reads as follows in the nodal form:
\[
\frac{1}{2a\tri t}\cdot (u^{n+1}_{j+1}-u^{n}_{j+1}+u^{n+1}_{j}-u^{n}_{j})+ 
\frac{1}{\tri x}\cdot (\theta (u^{n+1}_{j+1}-u^{n+1}_{j}) 
- (\theta-1) (u^{n}_{j+1} - u^{n}_{j})) = 0.
\]
As one can easily see, this scheme is consistent with its differential
equation. The stability polynomial is linear with complex coefficients, 
so we present it as a fraction. 
The reformulated stability problem, which we have to solve, is
\[
-2 \leq 
\frac{4a^2d^2 ( 2 \theta - 1)}{4a^2d^2 (\theta-1)^2 + \frac{1+\sin(\beta)}{1-\sin(\beta)} }
\leq 0, \; \forall \beta \not\in \frac{\pi}{2} \Z
\]
Since $t:= \tfrac{1+\sin(\beta)}{1-\sin(\beta)} \geq 0$, the right hand side inequality is equivalent to $\theta\leq \tfrac{1}{2}$. The left hand side is equivalent to $4a^2d^2 ( 2 \theta - 1) +2 (4a^2d^2 (\theta-1)^2 + t)) \geq 0$. Since $t\in [0,\infty)$, we have to show that 
$0 \leq 4a^2d^2 ( 2 \theta - 1) + 8a^2d^2 (\theta-1)^2 = 4a^2d^2(\theta^2 + (\theta-1)^2)$, what is true for all $d$. Of course, computations with CAD confirm this answer.

Thus, this scheme is unconditionally stable if $\theta \leq \tfrac{1}{2}$ and unstable otherwise.

\end{example}


\section{Examples for $\lambda$-wave equation}
\label{sExWave}

We consider a parametric equation $u_{tt} - \lambda^2 u_{xx} = 0$
with parameter $\lambda \not=0$ and its higher dimensional versions. We construct finite difference schemes for several different approximations and analyze them for stability.

\subsection{Conservative law with parametric time approximation} 

The presentation via the conservation law is
$\oint\limits_{\Gamma} \lambda^2 u_x dt + u_t dx = 0$.
We use trapezoid rule for both for the contour integral and spatial integral relations. For temporal integral relations, we employ the parametric difference with $\theta \in [0,1]$. 
We obtain the following system of difference equations:

\[
\begin{pmatrix}
\tri h \cdot (- T_x T_t+T_x +T_t -1) & \lambda^2 \tri t \cdot (T_x T_t - T_t - T_x +1) & 0 \\
0 & \frac{1}{2} \tri x \cdot (T_x+1) &  1-T_x \\
\tri t \cdot(\theta T_t +(1-\theta)) & 0  & 1- T_t
\end{pmatrix}
\bullet
\begin{pmatrix}
u_t\\
u_x\\
u
\end{pmatrix}
=0
\]

After the computation of Gr\"obner basis, we obtain the scheme

\begin{eqnarray*}
0 = \frac{1}{2 \tri t ^{2}}\cdot (u^{n+2}_{j}-2 u^{n+1}_{j}+u^{n}_{j}) - 
\frac{1}{2 \tri t ^{2}}\cdot (u^{n+2}_{j+2} -2 u^{n+1}_{j+2}+u^{n}_{j+2})  + \\
+ \frac{ \lambda ^{2}}{ \tri h ^{2}}\cdot
\Bigl(  \theta \cdot (u^{n+2}_{j+2}-2 u^{n+2}_{j+1}+u^{n+2}_{j})
-(2 \theta -1)\cdot (u^{n+1}_{j+2} - 2 u^{n+1}_{j+1} + u^{n+1}_{j})  + 
(\theta -1)\cdot (u^{n}_{j+2}-2  u^{n}_{j+1}+u^{n}_{j}) \Bigr)
\end{eqnarray*}


The stability polynomial of 2nd degree is rather complicated. However, 
factorization reveals a factor $g-1$. The other
factor is linear, but with complicated coefficients. We present
it as $g - \tfrac{c}{d'}$. Since both $c$ and $d'$ are complex
numbers, we compute the absolute value of them. Then, 
$||d'|| = (4\theta^2 d^4-1)\cdot(\cos(\beta_x)-1) -2$ and
$||c|| = ||d'|| - 4d^4(2\theta-1)(\cos(\beta_x)-1)$. 
\[
\text{Then, } 
||\frac{c}{d'}|| \leq 1 \Leftrightarrow 0 \leq 4 d^4 \sin(\beta_x/2)^2 \dfrac{(2\theta-1)}{(4\theta^2 d^4-1)\sin(\beta_x/2)^2 +1} \leq 1.
\]

Consider the left hand side inequality
\[
0 \leq 4 d^4 \sin(\beta_x/2)^2 \dfrac{(2\theta-1)}{(4\theta^2 d^4-1)\sin(\beta_x/2)^2 +1} \Leftrightarrow 0 \leq (2\theta-1)((4\theta^2 d^4-1)\sin(\beta_x/2)^2 +1)
\]

Since $4\theta^2 d^4 > 0 \Leftrightarrow 4\theta^2 d^4 -1 > -1 \geq -\dfrac{1}{\sin(\beta_x/2)^2}$, the second factor is always positive. Hence, the inequality is satisfied as soon as $\theta \geq \frac{1}{2}$.

The second inequality reads as
$4 d^4 \sin(\beta_x/2)^2 \dfrac{(2\theta-1)}{(4\theta^2 d^4-1)\sin(\beta_x/2)^2 +1} \leq 1$. Then,
\[
4\theta^2 d^4 \sin(\beta_x/2)^2 + 1 -\sin(\beta_x/2)^2 \geq (4 d^4 \sin(\beta_x/2)^2)(2\theta-1) \Leftrightarrow
\]

\[
\theta^2 + \frac{\cos(\beta_x/2)^2}{4 d^4 \sin(\beta_x/2)^2} \geq (2\theta-1) \Leftrightarrow (\theta-1)^2  + \frac{\cos(\beta_x/2)^2}{4 d^4 \sin(\beta_x/2)^2} \geq 0,
\]

what is always the case. Summarizing, we obtain that this scheme is unconditionally stable, if $\theta \geq \frac{1}{2}$ and unstable otherwise.

\subsection{Integral relations and 2nd order central approximations}
\label{goodWave}

Using direct 2nd order central approximations for both $t$ and $x$, we obtain
the following scheme:

\[
(u^{n+2}_{j+1}-2 u^{n+1}_{j+1}+u^{n}_{j+1}) - \lambda^{2} \frac{\tri t^{2}}{ \tri h^{2}}\cdot (u^{n+1}_{j+2}-2 u^{n+1}_{j+1}+u^{n+1}_{j}) = 0
\]

We denote $d:= \lambda \dfrac{\tri t}{\tri h}$, then the polynomial, describing the scheme is
\[
p = d^2 T_x^2 T_t-T_x T_t^2 +(-2 d^2+2)T_x T_t-Tx+d^2 T_t
= T_x (T_t-1)^2 - d^2 (T_x- 1)^2 T_t,
\] 
with the second expression being 
the convenient semi-factorized form

As usual, we use the stability morphism and simplifications. After them, the stability polynomial reads then as $g^2+ (-2 + 4d^2 \sin^2(a/2))g +1 = 0$.
Denote $b:=-1+2d^2 \sin^2(a/2)$, i.e. 
$g^2 + 2b g +1 =0$ and the roots are $b\pm\sqrt{b^2-1}$. 
If $b^2>1$, then one of the roots has modulus bigger, than one.
If $b^2=1$, the roots are $\pm 1$.
If $b^2<1$, the absolute value of both roots equals $b^2+1-b^2=1$.
Hence, $b^2\leq 1$, what is satisfied if and only if $d\leq 1$ that is
$\dfrac{\tri t}{\tri h} \leq \dfrac{1}{\lambda}$. The same condition is produced with the help of CAD in \textsc{Mathematica}.

This scheme is conditionally stable with the condition for the Courant number $d=\lambda \dfrac{\tri t}{\tri h}\leq 1$.

\subsection{Explicit integration for $t$ and trapezoid rule for $x$}

Using the explicit integration (that is, a backward difference) for $t$ and 
trapezoid rule for $x$, we get the following scheme.

\begin{eqnarray*}
\frac{1}{4 \tri t^{2}}\cdot (u^{n+2}_{j+2}-2 u^{n+1}_{j+2}+u^{n}_{j+2}
+ 2(u^{n+2}_{j+1}-2 u^{n+1}_{j+1}+ u^{n}_{j+1}) + u^{n+2}_{j}-2 u^{n+1}_{j}+u^{n}_{j}) - \\
- \frac{ \lambda^{2}}{ \tri h^{2}}\cdot (u^{n+2}_{j+2}-2 u^{n+2}_{j+1}+u^{n+2}_{j}) = 0
\end{eqnarray*}


The difference scheme polynomial is
\[
T_x^{2} T_t^{2}-2 T_x^{2} T_t+2 T_x T_t^{2}+ T_x^{2}-4 T_x T_t+ T_t^{2}+2 T_x-2 T_t+1 - \dfrac{4 \lambda^{2} \tri t^2}{ \tri h^{2}} (T_x^{2} T_t^{2}-2 T_x T_t^{2}+ T_t^{2})
\]

Denote $d^2:= \dfrac{4 \lambda^{2} \tri t^2}{ \tri h^{2}}$. After performing
substitutions, we obtain $g^2 - 2bg +b =0$, where $b = (1+d^2\tan^2(a))^{-1}$.
Its solutions are straightforward:  $g = b \pm \sqrt{b^2-b}$.
If $b^2-b>0$, we have $b>1$ and hence one root is too big.
If $b^2-b\leq 0$, the absolute value of a root is just $b^2+b-b^2=b$, what
is not bigger than $1$. And of course, $b \leq 1$ is satisfied, since
$b= 1+d^2\tan^2(a))^{-1}$. Hence, this scheme is unconditionally stable.
With the help of CAD and \textsc{Mathematica}, we arrive
to the same conclusion.

\subsection{Higher dimensional $\lambda$-wave equation}

One of the crucial advantages of our approach and its implementation is the scalability. That is, we employ the algorithms in the very general setting. They can be easily modified for the case of more functions (like $u$) involved. In particular, we are able to generate schemes and test them for stability in the higher-dimensional setting.

Consider the approach from Subsection \ref{goodWave} which led us to the
conditionally stable scheme. In what follows, we apply the same approximations
to all the spatial variables.\\

\textbf{Two spatial dimensions.} We have $u_{tt} - \lambda^2 (u_{xx} + u_{yy}) = 0$. The scheme, which we obtain is

\begin{eqnarray*}
0 = \frac{1}{ \tri t^{2}}\cdot (u^{n+2}_{j+1,k+1}-2 u^{n+1}_{j+1,k+1}+u^{n}_{j+1,k+1})- \\
- \frac{ \lambda^{2}}{ \tri x^{2}}\cdot  (u^{n+1}_{j+2,k+1} -2 u^{n+1}_{j+1,k+1} + u^{n+1}_{j,k+1}) - \frac{ \lambda^{2}}{ \tri y^{2}}\cdot (u^{n+1}_{j+1,k+2} -2u^{n+1}_{j+1,k+1}+u^{n+1}_{j+1,k} ) 
\end{eqnarray*}

In a semi-factorized form, the scheme looks as follows
\[
T_x T_y (T_t-1)^2 - d_x^2 \cdot (T_x-1)^2 T_y T_t  - d_y^2 \cdot T_x (T_y-1)^2 T_t=0.
\]

The stability polynomial in a simplified form is
\[
g^2 -2(d_x^2 \cos(\beta_x) + d_y^2 \cos(\beta_y) - d_x^2-d_y^2-2) \cdot g +1 = 0.
\]

Using CAD, we conclude, that this scheme is \textbf{conditionally stable} with the condition $d_x^2 + d_y^2 \leq 1$ for the Courant numbers $d_x:=\lambda \dfrac{\tri t}{\tri x}$, $d_y:=\lambda \dfrac{\tri t}{\tri y}$. \\

\textbf{Three spatial dimensions.} The corresponding equation is $u_{tt} - \lambda^2 (u_{xx} + u_{yy} + u_{zz}) = 0$.\\ 
The difference scheme is analogous to the two--dimensional one, in a semi-factorized form it has the following form (from which one easily deduces, how the scheme looks in yet higher dimensions):
\[
T_x T_y T_z (T_t-1)^2 - d_x^2 \cdot (T_x-1)^2 T_y T_z T_t - d_y^2 \cdot T_x (T_y-1)^2 T_z T_t - d_z^2 \cdot T_x T_y (T_z-1)^2 T_t = 0.
\]

Running CAD, we obtain, that this scheme, as its lower-dimensional analogues, is 
\textbf{conditionally stable} with the condition $d_x^2 + d_y^2 + d_z^2 \leq 1$ for the Courant numbers $d_x$, $d_y$ and $d_z:=\lambda \dfrac{\tri t}{\tri z}$.

\section{Dispersion Analysis}
\label{sDispersion}
\subsection{Continuous Dispersion}
Recall, that a \textbf{Fourier node} 
in n+1 dimensions is the function of the form
\[
e^{i( \langle k,x \rangle  - \w t)}, \;\; \langle k,x \rangle := \sum^n_{j=1} k_j x_j
\]

Respectively, in  1+1 dimensions it is just $e^{i(kx-\w t)}$. 
One obtains continuous dispersion from the given linear PDE by substituting Fourier nodes into the equation and deriving a relation $\w = \w (k)$ from the result. 
\begin{example}
For the equation $u_{tt} -  \lambda^2 u_{xx}=0$ we have
\[
0 = (\dfrac{\d}{\d t^2} - \lambda^2 \dfrac{\d}{\d x^2})e^{i(kx- \w t)} = 
-e^{i(kx-\w t)}\cdot(\w^2 - \lambda^2 k^2)
\]
Hence, $\w = \pm \lambda k$ is the continuous dispersion relation for
the $\lambda$--wave equation.
\end{example}

We can write down the action of partial derivatives on a Fourier mode.
Namely, 
\[
\dfrac{\d^a}{\d t^a} (e^{ i(\langle k,x \rangle - \w t)}) = (-i \w)^a e^{ i(\langle k,x \rangle - \w t)}  \text{ and }
\dfrac{\d^{b_j}}{\d x_j^{b_j}} (e^{ i(\langle k,x \rangle - \w t)}) = (i k_j)^{b_j} e^{ i(\langle k,x \rangle - \w t)}.
\]

Hence, the monomial in partial differentiations has its eigenvalue
\[
\dfrac{\d^a}{\d t^a} \prod^{n}_{j=1} \dfrac{\d^{b_j}}{\d x_j^{b_j}}  (e^{ i(\langle k,x \rangle - \w t)}) = (-i \w )^a \prod^{n}_{j=1} (i k_j)^{b_j}  \cdot  (e^{ i(\langle k,x \rangle - \w t)})
\]

Let us denote $F = e^{ i\langle k,x \rangle - i \w t)}$. Then $\partial^{\alpha} (F) = c(\alpha) \cdot F$, where $\alpha:=(a,b_1,\ldots,b_n) \in \N^{n+1}$. Extending this action by linearity to the ring of partial differentiations with constant coefficients $R=\K[\d_t,\d_{x_1},\ldots,\d_{x_n}]$, we are able to compute the eigenvalue of any polynomial from $R$ on $F$:
\[
P(F) = \sum_{\alpha} p_{\alpha} \partial^{\alpha} (F) = \bigl (\sum_{\alpha} p_{\alpha} c(\alpha) \bigr ) \cdot F.
\]

Then, the continuous dispersion relation is obtained by solving with respect to $\w$ the equation 
\[
\sum_{\alpha} p_{\alpha} c(\alpha) = 0, \;  p_{\alpha}\in \K, \;  c_{\alpha}\in \K(k_1,\ldots,k_n,\w ),
\]
which is called the \textbf{continuous dispersion equation} (CDE) for $P$.

\begin{example}
For $1+n$-dimensional heat equation $u_t - a^2 \cdot \sum^{n}_{j=1} u_{x_j x_j} = 0$
the continuous dispersion equation is 
\[
0 = - i\w -a^2  \sum^{n}_{j=1} i^2 k_j^2 \; \Longleftrightarrow \; \w = -i a^2 \sum^{n}_{j=1} k_j^2.
\]
\end{example}

\begin{example}
For $1+n$-dimensional modified $\lambda_i$-wave equation $u_{tt} - \sum^{n}_{j=1} \lambda_j^2 \cdot u_{x_j x_j} = 0$ 
the continuous dispersion equation is 
$w = \pm \sqrt{\sum_{j=1}^n \lambda_j^2 k^2_j} $.
\end{example}

\subsection{Discrete Dispersion}

In the discrete case, we consider a discrete Fourier node, 
corresponding to the grid point $(t_m,(x_{1})_{l_1}, \ldots, (x_{n})_{l_n})$,
\[
F^m_l = e^{i \langle k,x_{(l)} \rangle  - \w t_m}, \;\; \langle k,x_{(l)} \rangle := \sum^n_{j=1} k_j (x_{j})_{l_j}.
\]

One substitutes a discrete Fourier node into the difference scheme and derives
a relation $\w = \w (k)$ from the result. Let us write down the formula for the eigenvalue of a monomial:
\[
T_t^a \prod^{n}_{j=1} T_{x_j}^{b_j} (e^{i \langle k,x_{(l)} \rangle  - \w t_m}) = 
(e^{-i \w \tri t})^a \prod^{n}_{j=1} (e^{i k_j \tri x_j})^{b_j} \cdot (e^{i \langle k,x_{(l)} \rangle  - \w t_m}).
\]

As in the continuous case, we extend this action by linearity to polynomials. For a polynomial $P\in\K[T_t,T_{x_1},\ldots,T_{x_n}]$ one has
\[
P(F^m_l) = \sum_{\alpha} p_{\alpha} T^{\alpha} (F^m_l) = \bigl (\sum_{\alpha} p_{\alpha} c(\alpha) \bigr ) \cdot F^m_l,
\]
so we solve the \textbf{discrete dispersion equation} (DDE) for $P$,
\[
\sum_{\alpha} p_{\alpha} c(\alpha) = 0, \;  p_{\alpha}\in \K, \;  c_{\alpha}\in \K(\{k_j\},\w ),
\]
and obtain the discrete dispersion relation. Note, that in contrast to the continuous case, this relation is not of polynomial form in general.

Presenting discrete Fourier nodes via trigonometrical functions, we are able to compute discrete dispersion relations symbolically. We prefer not to use the 
de Moivre's formula, but to express dispersion relations in terms of sine and cosine of a single argument.

We work in the commutative ring $\C(\tri t, \tri x)[sin_t, cos_t, \{sin_{j}, cos_{j}\}]$ modulo the ideal $\langle \{ sin^2_{j}+cos^2_{j}-1 \}, sin^2_t+cos^2_t-1 \rangle$, where 
$cos_{j} :=  \cos(k \tri x_j)$, $cos_t := \cos(\w \tri t)$. Then,
\[
T_t^a \prod^{n}_{j=1} T_{x_j}^{b_j} (F^m_l) = 
(\cos_t - i \sin_t)^a \prod^{n}_{j=1} (\cos_j + i \sin_j)^{b_j} \cdot (F^m_l).
\]

\begin{example}
Consider the  $\lambda$-wave equation $u_{tt} -  \lambda^2 u_{xx}=0$
and the difference scheme 
\[
d^2 T_x^2 T_t-T_x T_t^2 +(-2 d^2+2)T_x T_t-T_x+d^2 T_t = 0,
\]
obtained with the $2^{nd}$ order central approximations for $x$ and $t$, we denote $d:= \lambda \frac{\tri t}{\tri x}$.

Performing computations, we obtain after simplification $d^2 \cos_x - cos_t+1-d^2 = 0$, that is $\cos(\w \tri t) = 1 - d^2(1-\cos(k \tri x))$. 
In the stability limit $d\to 1$, we have $\cos(\w \tri t) = \cos(k \tri x)$,
hence $w = \pm \frac{\tri x}{\tri t} k + 2 \pi m, m\in \Z$.
Since $d \to 1$ implies $\frac{\tri x}{\tri t} \to \lambda$, 
in the stability limit the discrete dispersion relation becomes
$\w = \pm \lambda k + 2\pi m$, where for $m=0$ we recover the continuous dispersion relation.
\end{example}

\section{Conclusion and Future Work}

The advantages of the methods we propose include, among other, 
their scalability and tendency towards automatization. Indeed,
we do not make distinction between classical types of PDE's
(hyperbolic, elliptic, parabolic). Thus these methods are
very general. Symbolic methods are able to generate automatically 
many difference schemes of standard linear PDE's, as
it was demonstrated in \cite{BGM} and by ourselves.

An important issue for the future research is
the (partial) algebraization of the consistency of a generated scheme
with the differential equation. Provided such a check, one could
work with general multi-parametric schemes, where the conditions
on parameters arise from the consistency check and the symbolic
stability approach.

We decided not to include the treatment of systems of linear PDE's in
this paper. However, we want to remark, that by the rewriting system
approach the number of the discretized equations is exactly the number
of PDE's one started with. By using Gr\"obner bases, we get in general
more equations, which, in turn, reveal the interplay between discretized
equations. Such interplay is not detected by the rewriting approach at
all; it seems to us that numerists just ignored such interplay. Hence,
this issue need to be investigated further.

Christian Dingler (TU Kaiserslautern, Germany) has been developing a new package
for \textsc{Singular} with \textsc{QEPCAD} as an engine for cylindrical
algebraic decomposition. This package extends the tools for the generation
of finite difference schemes to the cases of a single linear PDE and
of a system of linear PDE's. Another problem for the further research is
the generalization of von Neumann stability for systems, which is
clear only for some classes of equations.

The generalization of the methods we described goes in several different
directions: allowing variables coefficients and/or allowing nonlinearity
of expressions in functions. The generation of schemes in such settings
is still possible (by the rewriting system approach), but 
even the notion of stability, to the best of our knowledge, needs to be investigated
depending on the particular class of equations.

Very important question is the role of differential and difference Gr\"obner bases 
for nonlinear equations in the scheme generation and analysis. The recent works
\cite{BGM, BG09} show, that in some cases the systematic use of interplay
between equations can produce more unviversal (though, of course, more complicated)
schemes, which hints at big potential of these methods.

In connection with this, a new theory for infinite (but still constructively
approachable!) difference Gr\"obner bases, arising from the \textit{letterplace}
philosophy \cite{LL09}, is of great interest.

\section*{Acknowledgments}

The authors express their gratitude to V.~P.~Gerdt (JINR, Russia) and M.~Fr\"ohner (TU Cottbus, Germany) for their interest, discussions and suggestions during the work on this paper. 
We would also like to thank M.~Kauers (RISC, Linz, Austria), W.~Zulehner (J.~Kepler University of Linz, Austria) and A.~Klar (TU Kaiserslautern, Germany) for discussions on various topics around stability in this paper. We have learned many examples from the scripts and papers of colleagues, mentioned above. A special thanks goes to H.~Engl (Vienna, Austria) for his constructive critics, which helped to improve the presentation of the results.
 
The first author is grateful to the SFB F013 ``Numerical and Symbolic Scientific Computing'' of the Austrian FWF for partial financial support in 2005-2007.

\section*{Appendix. Tools for finite difference schemes in \textsc{Singular}}



\section{An introduciton to the system {\sc Singular} and its language}
\label{SingLang}

Next we describe shortly by examples how to read {\sc Singular} language and how 
to obtain and interpret the output \-- as far as it is used
to generate a difference scheme. More details are found in the {\sc Singular} manual at \texttt{www.singular.uni-kl.de}. 

\begin{itemize}
\item Defining an algebra. At first a domain for the computation has to be defined.
This is done by the following input:
\begin{quote}
\begin{verbatim}
ring R = (0,ro,K,dt,dh),(Tx,Tt),(c,dp);
\end{verbatim}
\end{quote}
It defines a polynomial ring, denoted by $R$, which represents an algebra over 
a prime field. The first brackets defines the coefficient field. It starts with a 
number for the characteristic, here $0$ stands for the prime field of 
rational numbers,
which is always used in our considerations. A list of names for parameters follows.
These parameters are constants \-- as step size in space or time ($dh, dt$) 
or constants of the equations ($ro, K$). Hence we have defined $\K$ 
as field of coefficients, here: rational functions over $\Q$ 
in the parameter $ro,K,dt,dh$.
The list of the second brackets corresponds to names of the variables $Tx,Tt$, 
which stand for commuting polynomial variables over the coefficient field $\K$ 
(we restrict us for this description to the commutative case only).
The names are almost arbitrary strings, hence names can be used to indicate 
the meaning \-- here the shift operators.
The string in the third brackets explains the monomial order used for
Gr\"obner basis computations: $dp$ \-- degree reverse lexicographical ordering
for polynomials. A small c at first place 
sort polynomial vectors by components first in descending order, i.e., 
${\bf e_1} > {\bf e_2}>\ldots$, then by the monomial ordering. 
\item Evaluating the constants:\\
Wanting to evaluate the constants in a resulting expression you have to
create a new ring, a mapping and substitutions:
\begin{quote}
\begin{verbatim}
ring r   = (0,dt,dh),(Tx,Tt,ro,K),(c,dp);
def  Ob1 = imap(R,Ob);
number num_ro = dh/2; number num_K = 4*dt^2 - dh^2; // constants
def  Ob2 = subst(Ob1,ro,num_ro);
     Ob2 = subst(Ob2,K,number_K);
\end{verbatim}
\end{quote}
Here  {\sl 'Ob'} is the name of the type 
(polynomial, matrix or ideal) you are considering. 
{\sl 'def Ob1'} creates an object of 
type of the right hand side. A new ring is introduced because a substitution
is defined only for variables.
\item Creating a matrix:
Starting with a linear system of PDE's, adding approximation rules, we end up 
with an extended system $\tilde{A}U=0$, see above.
We need only the matrix with entries in the ring $R$ of shift operators:
\begin{quote}\begin{verbatim}
matrix A[3][3] =
(-Tx*Tt^2+Tx), (Tx^2*Tt - Tt), 0 ,
0, (dh/2)*(Tx+1), 1-Tx,
(dt/2)*(Tt+1), 0, 1-Tt;
\end{verbatim}
\end{quote}
In the definition of a matrix you have to indicate row- and column-size, 
and on the right hand side just the list of polynomials.
\item Elimination of components:\\
We have to eliminate all components of the vector $U$ that stand 
for a differential operator, corresponding to Gauss-elimination over $R$
with the matrix. In this example the anonymous vector $U$ stands for
$(u_t,u_x,u)^t$. We want to produce with $A$ a row, having entries
only in the last component(s). This is done most 
efficiently by a Gr\"obner basis computation
of the module generated by the columns of the 
transposed matrix in the indicated monomial ordering.
The last non-zero component(s) of the first generator(s) 
correspond to the difference scheme, here:
\begin{quote}\begin{verbatim}
module M  = transpose(A);
module M1 = std(M); 
print(M1);
M1[3,1];
\end{verbatim}
\end{quote}
The command {\sl 'print'} returns only the first digits of the string corresponding 
to any entry. Type the indices of an entry in square brackets, to get 
it in full length. Here the resulting equation is returned as:
\begin{quote}\begin{verbatim}
(-dt)*Tx^2*Tt+(dh)*Tx*Tt^2+(2*dt-2*dh)*Tx*Tt+(dh)*Tx+(-dt)*Tt
\end{verbatim}
\end{quote}
\end{itemize}

\subsection{Tools for difference schemes}


The library \texttt{findifs.lib} has been created to automate numerous
processes, taking place while working with finite difference schemes.
An important role is played by sophisticated routines, transforming the
 different forms of objects into some classical ones. One can generate
complicated schemes and easily present them in, say, nodal form. At the
same time, one keeps the polynomial operator presentation, which is used
in e.g. stability analysis.

\begin{itemize}

\item  \texttt{decoef(P,n);}  P a poly, n a number.\\
Decomposes the poly P into summands with respect to the presence of a number n in the coefficients, returns an ideal in usually two generators.

Example:
\begin{verbatim}
  ring r = (0,dh,dt),(Tx,Tt),dp; 
  poly P = (4*dh^2-dt)*Tx^3*Tt + dt*dh*Tt^2 + dh*Tt; 
  P;
==> (4*dh^2-dt)*Tx^3*Tt+(dh*dt)*Tt^2+(dh)*Tt
  decoef(P,dt); 
==> _[1]=(4*dh^2)*Tx^3*Tt+(dh)*Tt   // the part, not containing dt
    _[2]=(-dt)*Tx^3*Tt+(dh*dt)*Tt^2 // the part which contains dt
  decoef(P,dh); 
==> _[1]=(-dt)*Tx^3*Tt               // the part, not containing dh
    _[2]=(4*dh^2)*Tx^3*Tt+(dh*dt)*Tt^2+(dh)*Tt
\end{verbatim}

\item \texttt{difpoly2texS,P[,Q])};  S an ideal, P and optional Q are lists.\\
Presents the difference scheme, given in the ideal S, in the nodal form.
The ideal S thought to be the result of \texttt{decoef}, list P contains
parameters, which will be controlled in order to remain in numerators. The optional
list Q contains polynomials, which will be added to the scheme (written in the function $u$) the part in terms of  a function $p$.

Example:
\begin{verbatim}
  ring r = (0,dh,dt,V),(Tx,Tt),dp; 
  poly M = (2*dh*Tx+dt)^2*(Tt-1) + V*Tt*Tx;
  M;
==> (4*dh^2)*Tx^2*Tt+(-4*dh^2)*Tx^2+(4*dh*dt+V)*Tx*Tt+(-4*dh*dt)*Tx+(dt^2)*Tt+(-dt^2)
  ideal I = decoef(M,dt);
  I;
==> I[1]=(4*dh^2)*Tx^2*Tt+(-4*dh^2)*Tx^2+(V)*Tx*Tt
    I[2]=(4*dh*dt)*Tx*Tt+(-4*dh*dt)*Tx+(dt^2)*Tt+(-dt^2)
  list L; L[1] = V;
  difpoly2tex(I,L);
==> \frac{1}{4 \tri t}\cdot (u^{n+1}_{j+2}-u^{n}_{j+2}+\frac{ \nu}{4 \tri h ^{2}} 
    u^{n+1}_{j+1})+ \frac{1}{4 \tri h}\cdot (u^{n+1}_{j+1}-u^{n}_{j+1}+
    \frac{ \tri t}{4 \tri h} u^{n+1}_{j}+\frac{- \tri t}{4 \tri h} u^{n}_{j})
\end{verbatim}

The last output, compiled with TeX, produces
\[
\frac{1}{4 \tri t}\cdot (u^{n+1}_{j+2}-u^{n}_{j+2}+\frac{ \nu}{4 \tri h ^{2}} u^{n+1}_{j+1})+ \frac{1}{4 \tri h}\cdot (u^{n+1}_{j+1}-u^{n}_{j+1}+\frac{ \tri t}{4 \tri h} u^{n+1}_{j}+\frac{- \tri t}{4 \tri h} u^{n}_{j}).
\]

Now let us illustrate the use of the optional list Q.
\begin{verbatim}
ring D = (0,ro,K,dt,dh),(Tx,Tt),(c,Dp);
poly U = (-K*dt)*Tx^2*Tt+(K*dt)*Tt;
poly P = (-2*ro*dh)*Tx*Tt+(2*ro*dh)*Tx;
list V; V[1] = K; V[2] = ro;
difpoly2tex(-U,V,-P);  // we want P to be in terms of p, and U in terms of u
==> \frac{K}{2 \tri h}\cdot (u^{n+1}_{j+2}-u^{n+1}_{j})+ \frac{ \rho}{ \tri t}
    \cdot (p^{n+1}_{j+1}-p^{n}_{j+1})
\end{verbatim}

That is, we have produced the nodal form of scheme for two functions $u$ and $p$:
\[
\frac{K}{2 \tri h}\cdot (u^{n+1}_{j+2}-u^{n+1}_{j})+ \frac{ \rho}{ \tri t}\cdot (p^{n+1}_{j+1}-p^{n}_{j+1}).
\]

\item \verb?exp2pt(P[,L]);? convert a polynomial M into the TeX format, in nodal form
\item \verb?mon2pt(P[,L]);? convert a monomial M into the TeX format, in nodal form
\item \verb?texcoef(n);?  converts the number n into TeX
\item \verb?npar(n);? search for 'n' among the parameters and returns its number
\item \verb?replace(s,what,with);? replaces in s all the substrings with a given string
\item \verb?xchange(w,a,b);? exchanges two substrings of a string 
\item \verb?par2tex(s);? converts special characters to TeX in s.
\end{itemize}

\nocite{F84}

\bibliographystyle{abbrv}
\bibliography{lit}
\end{document}